%% file: main.tex
\documentclass[format=acmsmall, review=true]{acmart}

\usepackage{acm-ec-17-arxiv}

\usepackage{booktabs} 

\usepackage[ruled]{algorithm2e} 

\SetAlFnt{\small}
\SetAlCapFnt{\small}
\SetAlCapNameFnt{\small}
\SetAlCapHSkip{0pt}
\IncMargin{-\parindent}

\usepackage{xifthen}

\usepackage{amsmath}
\usepackage{amsthm}
\usepackage{amsfonts}
\usepackage{amssymb}
\usepackage{amscd}
\usepackage{eurosym}
\usepackage{mathrsfs}
\usepackage{mathtools}
\usepackage{bm} 
\usepackage{bbm}

\usepackage{pgfplots}

\allowdisplaybreaks

\input{math_macro}

\input{agt_macro}

\newcommand{\mydef}[1]{\textbf{#1}}

\newcommand{\ignore}[1]{}

\usepackage[]{color-edits}
\addauthor{sw}{blue}
\addauthor{vs}{red}

\begin{document}

\title[Simple vs Optimal Mechanisms in Auctions with Convex Payments]{Simple vs Optimal Mechanisms in Auctions with Convex Payments}  
\author{Amy Greenwald}
\affiliation{%
  \institution{Brown University}
  \department{Computer Science}
  \city{Providence}
  \state{RI}
  \country{USA}
}
\author{Takehiro Oyakawa}
\affiliation{%
  \institution{Brown University}
  \department{Computer Science}
  \city{Providence}
  \state{RI}
  \country{USA}
}
\author{Vasilis Syrgkanis}
\affiliation{%
  \institution{Microsoft Research}
  \department{New England}
  \city{Cambridge}
  \state{MA}
  \country{USA}
}
\input{abstract}

\maketitle

\input{introduction}

\input{related}

\input{model}

\input{payments_revenue}

\input{rev_ub}

\input{prior_free}


\input{all_pay}

\input{experiments}

\input{conclusion}


\bibliographystyle{ACM-Reference-Format}
\bibliography{bibliography}

\input{appendix}

\end{document}

%% file: math_macro.tex
\DeclareMathOperator*{\Exp}{\mathbb{E}}

\newcommand{\nsReal}{\mathbb{R}}

%% file: agt_macro.tex
\newcommand{\bidders}[1][]{N}
\newcommand{\numbidders}[1][]{n}
\newcommand{\goods}[1][]{G}
\newcommand{\numgoods}[1][]{m}
\newcommand{\alloc}[1][]{x_{#1}}
\newcommand{\allocvec}[1][]{\mathbf{\alloc}}
\newcommand{\intalloc}[1][]{{\hat{\alloc}}_{#1}}
\newcommand{\intallocvec}[1][]{\mathbf{\intalloc}}

\newcommand{\pertallocvec}[1][]{\tilde{\mathbf{\alloc}}}
\newcommand{\val}[1][]{v_{#1}}
\newcommand{\valvec}[1][]{\mathbf{\val}_{#1}}
\newcommand{\valw}[1][]{w_{#1}}
\newcommand{\valz}[1][]{z_{#1}}

\newcommand{\valspace}[1][]{T} 

\newcommand{\bid}[1][]{b_{#1}}
\newcommand{\bidvec}[1][]{\mathbf{\bid}_{#1}}
\newcommand{\payment}[1][]{p_{#1}}
\newcommand{\paymentvec}[1][]{\mathbf{\payment}_{#1}}
\newcommand{\intpayment}[1][]{\hat{p}_{#1}}
\newcommand{\intpaymentvec}[1][]{\hat{\mathbf{\payment}}_{#1}}
\newcommand{\util}[1][]{u_{#1}}
\newcommand{\virval}[1][]{\varphi_{#1}}

\newcommand{\costfun}[1][]{c_{#1}}
\newcommand{\costfunvec}[1][]{\mathbf{\costfun}_{#1}}
\newcommand{\intcostfun}[1][]{\hat{c}_{#1}}
\newcommand{\intcostfunvec}[1][]{\hat{\mathbf{\costfun}}_{#1}}
\newcommand{\quantile}[1][]{q_{#1}}

\newcommand{\PH}[1][]{%
\mathcal{PH}\ifthenelse{\equal{#1}{}}{}{\text{-}{#1}}%
} %
\newcommand{\MPH}[1][]{%
\mathcal{MPH}\ifthenelse{\equal{#1}{}}{}{\text{-}{#1}}%
} %

%% file: abstract.tex
\begin{abstract}

We investigate approximately optimal mechanisms in settings where
bidders' utility functions are non-linear;
specifically, convex, with respect to payments (such settings arise, for instance, in procurement auctions for energy).
We provide constant factor
approximation guarantees for mechanisms that are independent
of bidders' private information (i.e., prior-free),
and for mechanisms that rely to an increasing extent on that information (i.e., detail free).
We also describe experiments, which show that for
randomly drawn monotone hazard rate distributions,
our mechanisms achieve at least 80\% of the optimal revenue,
on average.  Both our theoretical and experimental results
show that in the convex payment setting, it is desirable
to allocate across multiple bidders, rather than only to bidders
with the highest (virtual) value, as in the traditional
quasi-linear utility setting.

\end{abstract}

%% file: introduction.tex
\section{Introduction}

The trade off between simplicity and optimality in mechanism design
has been well-explored over the past decade. A long line of work has
addressed whether simple mechanisms can achieve approximately optimal
performance in single-dimensional environments
\cite{Myerson:Optimal/1981,bulow1996auctions,Hartline2009,Hartline2015,Roughgarden2016}. Typically,
simplicity requirements take the form of prior-freeness: the mechanism
should be prior-free, meaning it should not depend on any
distributional knowledge about participants' private information; or,
it should be detail-free, meaning it should only minimally depend on
it.
Such \emph{oblivious\/} designs lead to more robust
guarantees, as they do not heavily depend on modeling assumptions, or
on data collection to learn about the participants' private information.

All of the aforementioned work on simple, optimal auctions assumes
that bidders' utilities are quasi-linear with respect to payments,
i.e., $\util[i] = \val[i] \alloc[i] - \payment[i]$, where
$\val[i] > 0$ is $i$'s private value, $\alloc[i]$ is his allocation,
and $\payment[i]$ is his payment to the auctioneer.
A notable exception is the work of \citet{Fu2013}, who consider prior-independent auctions for risk-averse agents,
who are modelled by a very specific form of capped quasi-linear utilities.

In this paper, we investigate the problem of simple, optimal auction
design, assuming utilities of the form $\util[i] = \val[i] \alloc[i] -
\costfun[i](\payment[i])$, where $\costfun[i](\cdot)$ is a convex
function that characterizes the impact of $\payment[i]$ on $i$'s
utility function, which we can interpret as $i$'s \emph{perceived\/}%
\footnote{While perceived payments may be hallucinated, for example by risk-averse bidders,
perceived payments refer, more generally, to any payments made to anyone other than the auctioneer.}
payment, as it is distinct from $\payment[i]$, $i$'s \emph{actual\/}
payment to the auctioneer.
We focus primarily on payments of the form
$c_i(p_i) = p_i^d$, for some fixed exponent $d \ge 2$.

Our original motivation for convex payments stemmed from a reverse
auction design problem in the realm of renewable energy markets.
Consider a procurement auction in which a government with a fixed
budget is offering subsidies (in euros, say) to power companies in
exchange for a supply of renewable energy (in watts, say).
Suppose the power companies' utility functions take this form:
$\mathfrak{\util}_{i} = \alloc[i] - \costfun[i](\payment[i]) / \val[i]$, 
where $\alloc[i]$ is some fraction of the total budget in euros, and
$\payment[i]$ is some deliverable amount of power in watts.%
\footnote{Multiplying $\mathfrak{\util}_{i}$ by $\val[i]$ 
yields a more familiar utility function---that of the forward auction setting: 
$\val[i] \mathfrak{\util}_{i} = \util[i] = \val[i] \alloc[i] - \costfun[i](\payment[i])$, 
with utility measured in units of power, rather than money.}  
The value $\val[i]$ can be understood as a production rate,
in that the company can produce $\val[i]$ watts per euro.
The assumption that $\costfun[i] (\cdot)$ is convex reflects
the fact that energy production costs may not be linear; 
on the contrary, it may be the case that as more energy is produced,
further units become more expensive to produce.

Another problem which also fits into our framework is the problem of allocating 
a fixed block of advertising time to retailers during the Superbowl.
In this application, an advertiser's utility is calculated by 
converting its allocation, in time, into dollars via its private value 
(measured in dollars per time unit), and then subtracting the cost of production: 
$\util[i] = \val[i] \alloc[i] - \costfun[i](\payment[i])$.
Here again, production costs can be convex; for example, the
advertiser might have to borrow funds to enable production, and in so
doing, might be subject to a non-linear (and publicly disclosed)
interest rate.

In both of these problems, the auctioneer seeks an optimal auction:
i.e., one that maximizes its total expected revenue.  Specifically, in
the energy problem, the government's objective is to maximize the
amount of power produced, subject to its budget constraint.  In the
advertising problem, the television network is seeking to maximize its
revenue for selling a fixed block of advertising time during the
Superbowl.  These examples motivate our goal of searching for simple,
optimal auctions, assuming convex payments.

\ignore{
\vsdelete{We further motivate the setting by pointing out that convex
payments can be used to offset any uncertainty a bidder might have
about its surplus ($v_i x_i$), in both the forward and the reverse
auction settings.  That is, if a bidder is uncertain about either its
private value, or about whether the auctioneer will indeed allocate as
promised, convex payments can be used to hedge against this
uncertainty.  Hallucinating costs that are greater than anticipated,
and planning accordingly, is one way to model risk aversion.}
}

The departure from quasi-linear utilities presents both technical and
qualitative differences.
As a first, the optimal mechanism does not appear to have such an
intuitive, closed-form characterization, but rather is the outcome of
an ugly convex program (see Section \ref{sec:optimal-auction}).
Beyond this technical difficulty,
the Myerson characterization~\cite{Myerson:Optimal/1981}, in which the
optimal auction in a symmetric enviornment with regular distributions
allocates to bidders with the highest values above some reserve, is no
longer valid.
As our next example shows,%
\footnote{We present a rough sketch of the proof here, deferring a detailed exposition to Appendix \ref{sec:example}.}
any mechanism that allocates only to the highest-value bidders cannot
achieve a constant factor approximation.

\begin{example} 
We show that, if we insist on allocating to only bidders of the highest value, the resulting mechanism can be very suboptimal, with a suboptimality ratio that decays to zero as the number of bidders $N$ approaches infinity, at a rate of $1/N^{1/4}$. Consider the following distribution of values: the value of each bidder is either $1$ with probability $\frac{\log(N)}{\sqrt{N}}$, or $1-\epsilon$ with probability $1-\frac{\log(N)}{\sqrt{N}}$. 
As the number of bidders grows large, then with very high probability there will be approximately $\sqrt{N}$ bidders with value $1$, and approximately $N-\sqrt{N}$ bidders with value $1-\epsilon$. So the interim allocation of a highest-bidders-win auction is approximately: $1/\sqrt{N}$ for a bidder with value $1$, and $0$ for a bidder with value $1-\epsilon$.
It follows that the payment of a bidder with value $1$ is
$\sqrt{xv} = \sqrt{x(1)} \approx 1/N^{1/4}$ (see Section~\ref{sec:payments}), while the payment of a bidder with value $1-\epsilon$ is approximately $0$. Thus, the expected payment of a single bidder is approximately $\left( \frac{1}{\sqrt{N}} \right) \left( \frac{1}{N^{1/4}} \right)$, and the total expected revenue is $N$ times this quantity, which is $N^{1/4}$.

On the other hand, if we instead allocate to all bidders uniformly at random (as long as they bid at least $1-\epsilon$), then we can charge everyone $\sqrt{xv} = \sqrt{\frac{1}{N} (1 - \epsilon)}$, leading to a total expected revenue of $\sqrt{N (1-\epsilon)}$. As $\epsilon\rightarrow 0$, the ratio of this latter mechanism to the former is $O(N^{1/4})$. 
\end{example}

This finding is not specific to this pathological example.  Experimentally, for many value distributions, we find that mechanisms that allocate only to bidders with the highest value, although they perform well in quasi-linear settings, can perform very poorly in the convex payment setting.

\paragraph{Contributions} Despite these technical and qualitative difficulties, we show that simple mechanisms, which are prior-free or detail-free, can achieve very good worst-case approximation ratios, relative to the optimal revenue. We also experiment with these simple mechanisms, as well as related mechanisms that were not so amenable to analysis, and find that they perform very close to optimally for a wide variety of distributions on values. 

More concretely, our main theoretical results hold when payment functions take the form $c_i(p_i) = p_i^d$, for $d \geq 2$, and when the distribution of values satisfies the Monotone Hazard Rate condition. For this setting, these results can be summarized as follows:

\begin{itemize}
\item We characterize an upper bound on the revenue of the optimal auction by finding a closed-form solution to a convex program that upper bounds the optimal revenue.
  
\item We show that a mechanism which sets a reserve price by drawing randomly from the distribution of values, and then allocates uniformly to all bidders above this reserve, is a constant factor approximation to the optimal mechanism.
This result implies a prior-free, constant factor approximation auction: randomly pick one bidder, and then allocate uniformly to all bidders who values fall above that of this price-setting bidder.
  
\item If the auctioneer knows the median of the value distribution, then a better approximation guarantee can be achieved by setting the reserve to be this median. Further optimization of the reserve price, tailored to the specific form of the convex payment function, \emph{assuming it is known}, leads to even better approximation ratios.
  
\item Even if the auctioneer does not know the exact form of the payment functions, but as long as they are of the form $p_i^d$, for some $d \geq 2$ which is known to the bidders, then a simple all-pay auction which allocates the good uniformly to the top quarter of the bidders, yields a constant factor approximation. This auction, while not truthful, is independent of the distribution of values and of the exponent $d$ in the payment function.  
\end{itemize}

In addition, we experiment with the performance of the aforementioned proposed auctions, as well as other auctions that look qualitatively similar. In our experiments, we draw from random distributions of values, and then we compute the revenue of each of auction, as well as the optimal revenue, for up to $30$ bidders. As a necessary step, we also show how one can formulate the optimal revenue computation as a convex program with the number of variables equal to the support of the distribution of values (for finite support distributions). We find that auctions that are worst-case constant factor approximations achieve on average at least 80\% of optimal, assuming more than $4$ bidders, thereby achieving much better guarantees experimentally than in theory.

%% file: related.tex
\paragraph{Related Work}

\citet{Vickrey:Counterspeculation/1961} showed that 
auctions in which the highest bidder wins and pays the second-highest bid 
incentivize bidders to bid truthfully.
\citet{Myerson:Optimal/1981} showed that in the
single-parameter setting, with the usual quasi-linear utility function
involving linear payments, total expected revenue is maximized by a
Vickrey auction with reserve prices.  Our setting is not captured by
Myerson's classic characterization because perceived payments in our
model are not equivalent to actual payments.

\citet{bulow1996auctions}, \citet{Hartline2009}, and \citet{Roughgarden2016} study simple prior independent mechanisms. The results in \citet{bulow1996auctions} show that running a simple second-price auction is a $(n-1)/n$ approximation in symmetric settings with quasi-linear utilities. \citet{Hartline2009} extend this analysis to obtain similar results in asymmetric settings. The result of \citet{bulow1996auctions} can also be phrased as obtaining a constant factor approximation by running a second-price auction with a reserve drawn randomly from the distribution of values. Our constant factor prior-independent result stems from this intuition. However, as we show in the introduction, allocating only to the highest bidders can be very suboptimal. Hence, we modify our mechanism to allocate uniformly at random to all bidders that surpass the reserve. We show that this simple modification is crucial to obtaining constant factor approximations in the convex payment setting.

The technical difficulties that arise in our setting are similar in
spirit to the ones faced by \cite{Pai2014} when designing
optimal auctions for budget-constrained bidders.
If $\costfun[i](\payment[i]) = \payment[i]^k$ for some $k \gg 0$,
then $\util[i] = \val[i] \alloc[i] - \payment[i]^k$ quickly approaches $-\infty$
if $\payment[i] > \left( \val[i] \alloc[i] \right)^{1/k}$.  
Thus, we can interpret a utility
function with convex perceived payments as a continuous approximation of that of
a budget-limited agent whose utility is $-\infty$ whenever her payment
exceeds her budget.

Our model also has strong connections with the literature on optimal auctions for risk-averse buyers (see \cite{Maskin1984}), since a concave utility function can be interpreted as a form of risk aversion.

Translating a reverse auction into a direct auction by
multiplying utility by the private parameter $\val[i]$ was previously
proposed in the literature on optimal contests 
(see, for example, \cite{Chawla2012,Dipalantino2009}).

Procuring services subject to a budget constraint is also the subject
of the literature on budget-feasible mechanisms initiated by
\cite{Singer2014}.  However, in this literature, the service of
each seller is fixed and the utility of the buyer is a combinatorial
function of the set of sellers the buyer picks.  In our setting, each
seller can produce a different level of service 
by incurring a different cost, so the buyer picks not only a
set of sellers, but a level of service that each seller should provide
as well.  This renders the two models incomparable.

Settings where bidders' utilities decrease 
at least as quickly as payments increase 
have been studied by \cite{sakurai2008beyond}, in the context of strategy-proof environments.

Our model is also related to the parameterized supply bidding game of \cite{Johari2011}, 
where firms play a game in which they submit a single-parameter family of supply functions. 
There, a non-truthful mechanism decides how much each firm is asked to produce.
Here, we consider the design of truthful mechanisms and we study a different objective, 
but we also restrict attention to single-parameter families of supply functions $\costfun[i](\payment[i])$.

Optimal mechanism design with non-linear preferences (mostly budget constraints) 
has been analyzed by several papers 
and we refer the reader to Chapter 8 of \cite{Hartline2015} 
for a recent survey on the state of the art. 
For the case of non-linear preferences that we analyze, 
no closed form solution of the optimal mechanism is known and 
thereby deriving intuitive approximations seems like a good alternative.  
In principle, some formulations of our problem 
can be solved using Border's characterization of interim feasible outcomes \cite{Border1991} 
and an ellipsoid-style algorithm with a separation oracle. 
Even more generally, we can apply the algorithmic approach of \cite{Cai2013} 
for computing the optimal mechanism, which again is based on an ellipsoid-style algorithm. 
However, such mechanisms tend to be computationally expensive and 
do not yield closed form characterizations of good mechanisms. 
Here, we seek fast allocation heuristics with potential economic justification, 
such as virtual-value-based maximization. 
Virtual-value-maximizing approximations to optimal auction design were also studied 
recently by \cite{alaei2013simple} in the context of multi-dimensional mechanism design, 
and from a worst-case point of view. 

%% file: model.tex
\section{Model and Preliminaries}

There is one seller who would like to sell one unit-sized divisible good,
and there are $\numbidders$ bidders that would like to buy as much of it as possible. Each bidder $i \in \bidders = \{ 1, \dots, \numbidders \}$ 
has a private value
for the good in its entirety.
Each value $\val[i]$ is drawn independently from an atomless distribution $F$, with continuous probability density $f$ 
that is strictly positive on the support, which is the closed interval $\valspace[i] = [0, \bar{\val}]$.
Let $\valvec = \left( \val[1], \dots, \val[\numbidders] \right) \in \valspace^n$ be a 
sample value vector, drawn from distribution $F^n$.

Given a vector of reports 
$\bidvec = \left( \bid[1], \dots, \bid[\numbidders] \right) \in \nsReal^{\numbidders}$,
with $\bid[i] \in \valspace[i]$, for all $i \in \bidders$,
a mechanism produces an allocation rule
$\allocvec ( \bidvec ) \in [0, 1]^{\numbidders}$ 
together with a payment rule
$\paymentvec ( \bidvec ) \in \nsReal_{\ge 0}$,
where bidder $i$'s (actual) payment to the seller is $\payment[i] (\bidvec)$.
For vectors such as $\bidvec$, 
we use the notation $\bidvec = (\bid[i], \bidvec[-i])$ 
to emphasize the distinction between bidder $i$'s role in the auction,
and all other bidders $\bidders \setminus \{ i \}$.

Let the utility of each bidder $i$ be 
\begin{equation}
\label{eq:utilityFunction}
\util[i] (\bid[i], \bidvec[-i]) 
=
\val[i] \alloc[i] (\bid[i], \bidvec[-i])
-
\costfun[i] \left( \payment[i] (\bid[i], \bidvec[-i]) \right)
,\end{equation}
where $\costfun[i] : \nsReal_{\ge 0} \rightarrow \nsReal_{\ge 0}$ is a 
convex payment function.
For readability, we often write
$\costfun[i] (\bid[i], \bidvec[-i])$
instead of 
$\costfun[i] \left( \payment[i] (\bid[i], \bidvec[-i]) \right)$.
This is bidder $i$'s \emph{perceived\/} payment,
beyond his actual payment $\payment[i] (\bidvec)$,
which only includes payments made to the auctioneer.
Similar to the payment rule, we refer to 
$\costfunvec (\bidvec) \in \nsReal_{\ge 0}^{\numbidders}$,
comprised of variables 
$\costfun[i] (\bid[i], \bidvec[-i])$,
as the perceived payment rule.

\subsection{Constraints}

Next, we formalize the constraints we impose on an optimal auction design.
Because we restrict our attention to incentive compatible auctions,
where it is optimal to bid truthfully, we write, for example, 
$\costfun[i] ( \val[i], \valvec[-i] )$
instead of $\costfun[i] ( \bid[i], \bidvec[-i] )$.

A mechanism is called \mydef{incentive compatible} (IC) if each bidder maximizes her utility
by reporting truthfully (i.e., $\bid[i] = \val[i]$):
$\forall i \in \bidders$,
$\forall \val[i], \valw[i] \in \valspace[i]$,
and $\forall \valvec[-i] \in \valspace[-i]$,
$\val[i] \alloc[i] ( \val[i], \valvec[-i] ) - \costfun[i] ( \val[i], \valvec[-i] )
\ge  
\val[i] \alloc[i] ( \valw[i], \valvec[-i] ) - \costfun[i] ( \valw[i], \valvec[-i] )$.
\mydef{Individual rationality} (IR) ensures that bidders have non-negative utilities:
$\forall i \in \bidders$, $\forall \val[i] \in \valspace[i]$, and $\forall \valvec[-i] \in \valspace[-i]$,
$\val[i] \alloc[i] ( \val[i], \valvec[-i] ) - \costfun[i] ( \val[i], \valvec[-i] )
\ge 
0$.

Next, we define IC and IR in expectation (with respect to $F_{-i}$).
We introduce \mydef{interim allocation} and \mydef{interim perceived payment}
variables, respectively:
$\intalloc[i] ( \val[i] ) \equiv \intalloc[i] ( \val[i], \cdot ) =
\Exp_{\valvec[-i]} \left[ \alloc[i] ( \val[i], \valvec[-i] ) \right]$ and
$\intcostfun[i] ( \val[i] ) \equiv \intcostfun[i] ( \val[i], \cdot ) =
\Exp_{\valvec[-i]} \left[ \costfun[i] ( \val[i], \valvec[-i] ) \right]$.
These variables comprise the interim allocation
and perceived payment rules,
$\intallocvec (\valvec) \in [0,1]^{\numbidders}$ and
$\intcostfunvec  (\valvec) \in \mathbb{R}_{\ge 0}^{\numbidders}$.

We call a mechanism \mydef{Bayesian incentive compatible} (BIC) if utility is maximized by
truthful reports in expectation:
$\forall i \in \bidders$ and $\forall \val[i], \valw[i] \in \valspace[i]$,
$\val[i] \intalloc[i] ( \val[i] ) - \intcostfun[i] ( \val[i] )
\ge  
\val[i] \intalloc[i] ( \valw[i] ) - \intcostfun[i] ( \valw[i] )$.
\mydef{Interim individual rationality} (IIR) insists on non-negative utilities in expectation:
$\forall i \in \bidders$ and $\forall \val[i] \in \valspace[i]$,
$\val[i] \intalloc[i] ( \val[i] ) - \intcostfun[i] ( \val[i] )
\ge 
0$.

We say a mechanism is \mydef{ex-post feasible} (XP) if it never overallocates:
$\forall \valvec \in \valspace^n$,
$\sum_{i=1}^{\numbidders} \alloc[i] ( \val[i], \valvec[-i] ) 
\le 
1$.
Finally, we require that
$0 
\le 
\alloc[i] ( \val[i], \valvec[-i] ), \intalloc[i] ( \val[i] ) 
\le 1$,
$\forall i \in \bidders$,
$\forall \val[i] \in \valspace[i]$ and $\forall \valvec[-i] \in \valspace^{n-1}$.

The goal of the auctioneer is to maximize \mydef{expected revenue}, which is equal to: $\Exp_{\valvec}\left[\sum_{i\in N} p_i(\valvec)\right]$. If we define with $\intpayment[i] (\val[i])=\Exp_{\valvec[-i]}\left[p_i(\val[i],\valvec[-i])\right]$, then the goal of the auctioneer is to optimize $\sum_{i\in N}\Exp_{\val[i]}\left[\intpayment[i] (\val[i])\right]$.

\subsection{Distributions and Properties}

We introduce some useful notation and terminolgy with respect to properties of the distribution $F$. For any distribution $F$, let $\quantile(\val) = 1 - F(\val)$, be the quantile function, and let $\val(\quantile) = \quantile^{-1}(\cdot)$, be the inverse quantile function. The quantile of a value $\val$ is the probability that a random draw from distribution $F$ exceeds $\val$.
Observe that
quantiles are distributed uniformly on $[0, 1]$.

In the usual quasi-linear setting, $R(\quantile) = \val(\quantile) \quantile = \val(\quantile) (1 - F(\val(\quantile))$ is the \mydef{revenue} function, assuming distribution $F$.
Intuitively, $R(q)$ is the expected revenue of a seller who posts a reserve price of $\val(\quantile)$ (i.e., one that is surpassed with probability $\quantile$).
In the convex payment setting, this function is the \mydef{perceived revenue} function.
Since, $F$ is atomless with support $[0,\bar{\val}]$, $R(0) = R(1) = 0$.

Finally, let $q^* \in \arg\max_{q\in [0,1]} R(q)$ be the quantile corresponding to optimal revenue, which we will refer to the \mydef{monopoly quantile}. Further, let $\eta = v(q^*)$ be the posted price corresponding to the monopoly quantile, which we will refer to as the \mydef{monopoly reserve}. We will also denote by $\kappa = v(1/2)$ the median of the distribution, and by $\mu = E_{v\sim F}[v]$, the mean.

We will be looking at two standard classes of distributions. The smaller class is that of monotone hazard rate (MHR) distributions, which require that $h(\val) = f(\val) / (1 - F(\val))$ be monotone non-decreasing. The larger class is that of regular distributions, which require that $R(q)$ be a concave function, or equivalently $\phi(q) = R'(q)= v(q)-\frac{1-F(v(q))}{f(v(q))}$, be monotone decreasing. Since $\phi(q)=v(q) - \frac{1}{h(v(q))}$, it is easy to see that an MHR distribution is also regular.

We now state two lemmas describing bounds on revenue curves, 
depending on the assumptions made on the 
distributions values are drawn from.

\begin{lemma}[\cite{dhangwatnotai2015revenue}]
\label{lem:dhangwatnotai}
For any MHR distribution,
\begin{equation}
R(q^*) \ge \frac{\mu}{e}
.\end{equation}
\end{lemma}

\begin{lemma}[\cite{samuel1984comparison}]
\label{lem:samuel}
For any regular distribution, $R(q^*) \le \kappa$.
\end{lemma}
\begin{proof}
\begin{equation}
\frac{R (q^*)}{2} 
\le 
R \left( \frac{1}{2} \right)
= 
v \left( \frac{1}{2} \right) F \left( 1 - \val \left( \frac{1}{2} \right) \right)
= 
\frac{\kappa}{2}
.\end{equation}
\end{proof}

%% file: payments_revenue.tex
\subsection{Allocation, Payments, and Revenue}
\label{sec:payments}

Myerson showed that for a mechanism to satisfy BIC, IIR and XP, four conditions needed to hold.
Among them, he showed that there is a specific payment formula that must be used.
We restate his result below, adapted to the convex payment setting.
For simplicity, we assume $u(0) = 0$, as will be the case for all our mechanisms.
Although we state this theorem for the BIC/IIR/XP setting, 
we note that it is easily extendable to the (robust) IC/IR/XP setting.


\begin{theorem}[\cite{Myerson:Optimal/1981}]
\label{thm:myersonIcIrPaymentFormulaMonotone}
Assuming a convex payment function,
a mechanism is BIC, IIR and XP if and only if the following conditions hold:
\begin{itemize}
\item The allocation rule is monotone:
\begin{equation}
\intalloc[i] (\val[i]) \ge \intalloc[i] (\valw[i]),
\quad
\forall i \in \bidders,
\forall \val[i] \ge \valw[i] \in \valspace[i]
,\end{equation}
\item (Perceived) payments satisfy the following condition:
\begin{equation}
\val[i] \intalloc[i] (\val[i]) - \intcostfun[i] (\val[i])
=\int_{0}^{\val[i]} \intalloc[i] (\valz[i]) \, \mathrm{d}\valz[i],
\quad
\forall i \in \bidders,
\forall \val[i] \in \valspace[i]
,\end{equation}
\end{itemize}
\end{theorem}

Myerson showed that the total expected revenue of such a mechanism
can be described using virtual values.
We restate his findings, adapted to the convex payment setting.

Define the \mydef{virtual value} function as follows:
\begin{equation}
\virval[i] (\val[i])
=
\val[i]
-
\frac{1 - F_i (\val[i])}{f_i (\val[i])} = R'(q(v_i))
.\end{equation}

\begin{theorem}[\cite{Myerson:Optimal/1981}]
\label{thm:myersonVirVal}
Assuming a convex payment function,
the total expected perceived payment of a BIC, IIR, and XP mechanism is:
\begin{equation}
\sum_{i \in \bidders} \Exp_{\val[i] \sim F_i} \left[ \intcostfun[i] (\val[i]) \right]
=
\sum_{i \in \bidders} \Exp_{\val[i] \sim F_i} \left[ \virval[i] (\val[i]) \intalloc[i] (\val[i]) \right] \, =
\sum_{i \in \bidders} \Exp_{\val[i] \sim F_i} \left[ R'(\quantile(\val[i]))\cdot \intalloc[i] (\val[i]) \right]
\end{equation}
\end{theorem}

In the traditional quasi-linear setting,
when the convex payment function is the identity function,
the fact that virtual surplus is equivalent to revenue tells us that in
order to maximize revenue, 
the good should be allocated to among bidders with the highest non-negative virtual values.

\ignore{
\vsdelete{\cite{bulow1989simple} showed this as well, using revenue curves.
The derivative of the revenue curve with respect to $\quantile[i]$ 
is the marginal revenue curve, 
and is equivalent to the virtual value function, $\virval[i]$,
so another way to interpret the mechanism Myerson proposed is to
allocate to the bidder with the highest positive marginal revenue.}
}

However, in the convex payment setting, the actual revenue of the auctioneer is not pinned down by Myerson's theorem; only perceived revenue is.
That is, Myerson's characterization only pins down the expected perceived payment of a bidder, conditional on that bidder's value,
i.e., $\intcostfun[i](\val[i]) = \Exp_{\valvec[-i]}\left[c_i(p_i(\val[i],\valvec[-i]))\right]$.
But the auctioneer is trying to maximize $\Exp_{\valvec}\left[\sum_i p_i(\valvec)\right]$,
which does not have any obvious interpretation using Myerson's characterization of perceived payments.

Nevertheless, we now prove a lemma that allows us to make such a connection.
We show that for all payment rules $p_i(\valvec)$ that might not be independent of $\valvec[-i]$,
there exists a corresponding payment rule $h_i(\val[i])$ that is independent of $\valvec[-i]$, and that ensures at least as much revenue.

Consider a (possibly randomized) auction, Auction $A$, where $\payment[i]^A
(\val[i], \valvec[-i], r)$ denotes bidder $i$'s payment in auction
$A$.  (Here, $r$ is the outcome of some randomization device.)  We
define another auction, Auction $B$, with payment rule
$\payment[i]^B (\val[i], \valvec[-i],r) = h_i (\val[i])$ for some
function $h_i(\val[i])$ that depends only on $\val[i]$.  More
specifically,
$h_i (\val[i])
= \costfun[i]^{-1} \left( \Exp_{\valvec[-i], r} \left[
\costfun[i] \left( \payment[i]^A (\val[i], \valvec[-i], r) \right) \right] \right)$.

\begin{lemma}
\label{lem:feasibleAllocAB}
\label{lem:revAB}
An arbitrary allocation $\allocvec \in [0,1]^{\numbidders}$,
together with the corresponding payment rule $\intpaymentvec^A$ or $\intpaymentvec^B$,
satisfies BIC, BIR, and XP for Auction $A$ if and only if
it satisfies BIC, BIR, and XP for Auction $B$. Moreover, Auction $B$'s total expected revenue is at least that of Auction $A$'s.
\end{lemma}

\begin{proof}
We begin with the first part. Observe that
the proposed payment rule preserves interim bidder payments, and hence, utilities. More concretely: 
\begin{equation}
\intcostfun[i]^B (\val[i]) = \Exp_{\valvec[-i]}  \left[ \costfun[i] \left( \payment[i]^B (\val[i], \valvec[-i]) \right) \right] = 
 \costfun[i] \left( h_i (\val[i]) \right) =
  \costfun[i] \left(\costfun[i]^{-1} \left( \Exp_{\valvec[-i], r} \left[
    \costfun[i] \left( \payment[i]^A (\val[i], \valvec[-i], r) \right) \right] \right)\right)=\intcostfun[i]^A (\val[i]).
\end{equation}
Hence, for all bidders $i \in \bidders$ and values $\val[i], \valw[i] \in \valspace[i]$,
$\val[i] \intalloc[i] ( \val[i] ) - \intcostfun[i]^B (\val[i])
\ge
\val[i] \intalloc[i] ( \valw[i] ) - \intcostfun[i]^B (\valw[i])$
iff
$\val[i] \intalloc[i] ( \val[i] ) - \intcostfun[i]^A (\val[i])
\ge
\val[i] \intalloc[i] ( \valw[i] ) - \intcostfun[i]^A (\valw[i])$
and
$\val[i] \intalloc[i] ( \val[i] ) - \intcostfun[i]^B (\val[i]) \ge 0$
iff
$\val[i] \intalloc[i] ( \val[i] ) - \intcostfun[i]^A (\val[i]) \ge 0$.

Now we prove the second part. Since $\costfun[i](\cdot)$ is convex, 
by Jensen's inequality,
\begin{equation}
h_i (\val[i])
=
\costfun[i]^{-1} \left( \Exp_{\valvec[-i], r} \left[ \costfun[i] \left( \payment[i]^A (\val[i], \valvec[-i], r) \right) \right] \right)
\allowbreak
\geq 
\allowbreak
\Exp_{\valvec[-i], r} \allowbreak \left[ \costfun[i]^{-1} \left( \costfun[i] \left( \payment[i]^A (\val[i], \valvec[-i], r) \right) \right) \right]
=
\allowbreak
\Exp_{\valvec[-i], r} \allowbreak \left[ \payment[i]^A (\val[i], \valvec[-i], r) \right].
\end{equation}
In other words, payments in Auction $B$ can only exceed those of Auction $A$.
Therefore, the total expected revenue of Auction $B$ is at least that of Auction $A$.
\end{proof}

This lemma establishes that in our search for an optimal auction, it
suffices to restrict our attention to auctions like Auction $B$ in
which payments are (deterministic) function of value alone:

\begin{corollary}
\label{cor:myersonPayment}
Any revenue maximizing mechanism that is BIC, IIR and XP also satisfies:
\begin{equation}
\label{eqn:det-payment}
\intpayment[i] (\val[i]) = \costfun[i]^{-1} \left( \intcostfun[i] (\val[i])  \right) = \costfun[i]^{-1} \left(\val[i] \intalloc[i] (\val[i]) -\int_{0}^{\val[i]} \intalloc[i] (\valz[i]) \, \mathrm{d}\valz[i]\right),
\quad
\forall i \in \bidders,
\forall \val[i] \in \valspace[i]
.\end{equation}
\end{corollary}

We will see in later sections that Equation~\ref{eqn:det-payment} leads to an easy calculation of the revenue of the optimal mechanism. However, this connection still does not lead to a nice characterization of the optimal mechanism, which is the outcome of a convex optimization problem, and can be arbitrarily complex. In the next sections, we will propose simple mechanisms that are approximately optimal.

%% file: rev_ub.tex
\section{Upper Bound on Optimal Revenue}

By individual rationality, we can upper-bound the total expected revenue
of a mechanism by
\begin{equation}
\sum_{i \in \bidders} \Exp_{\valvec \sim F} 
\left[
\payment[i] (\valvec)
\right]\leq \sum_{i \in \bidders} \Exp_{\valvec \sim F} 
\left[ 
\costfun[i]^{-1} \left( \val[i] \alloc[i] (\valvec) \right) 
\right]
.\end{equation}
We call the quantity on the right hand side, \mydef{pseudo-surplus}, and describe how to maximize it in Lemma~\ref{lem:solProgC}.

\begin{lemma}
\label{lem:solProgC}
\label{LEM:SOLPROGC}
When, $\forall i \in \bidders$, 
$\costfun[i] = \payment[i]^d$, 
where $d > 1$,
the allocations that maximize pseudo-surplus are given by 
\begin{equation}
\alloc[i] (\val[i], \valvec[-i])
=
\frac{ \val[i]^{1 / (d-1)} }{ \sum_{j \in \bidders} \val[j]^{1 / (d-1)} }
,
\quad
\forall i \in \bidders,
\forall \valvec \in \valspace^{\numbidders}
.\end{equation}
\end{lemma}
The proof of Lemma \ref{lem:solProgC} is a straightforward
application of the equi-marginal principle
\cite{gossen1854entwickelung},
and we explain it in full in Appendix~\ref{sec:progCproof}.

Using pseudo-surplus and Lemma~\ref{lem:solProgC},
we now give an upper bound to the total expected revenue 
that can be generated in a BIC/IIR mechanism, $OPT$,
which is described in terms of 
the number of bidders, their values and distributions, and $d$.
\begin{lemma}
\label{lem:optUbByPseudoSurplus}
\label{cor:optUbByPseudoSurplusFordtwo}
When, $\forall i \in \bidders$, 
$\costfun[i] = \payment[i]^d$, 
where $d \geq 2$:
\begin{align}
\label{eq:optUbByPseudoSurplus}
OPT 
&\le 
\numbidders \left( \frac{\mu}{\numbidders} \right)^{1/d},
\end{align}
where $\mu = \Exp_{v\sim F} [v]$.
\end{lemma}

\begin{proof}
Starting with the optimal solution to pseudo-surplus,
given by Lemma~\ref{lem:solProgC},
we can upper-bound $OPT$ as follows:
\begin{align*}
OPT
\leq~& 
\Exp_{\valvec} \left[
\sum_{i \in \bidders} \left( v_i \cdot \frac{v_i^{1/(d-1)}}{\sum_j v_j^{1/(d-1)}} \right)^{1/d}
\right] 
= 
\Exp_{\valvec} \left[
\sum_{i \in \bidders} v_i^{1/(d-1)} \frac{1}{\left( \sum_j v_j^{1/(d-1)} \right)^{1/d}} 
\right]
= 
\Exp_{\valvec} \left[
\left( \sum_{i \in \bidders} v_i^{1/(d-1)} \right)^{\frac{d-1}{d}}
\right]
.\end{align*}
Observe that $f(x) = x^{(d-1)/d}$ is a concave function for $d > 1$, 
thus $E[f(X)]\leq f(E[X])$. 
Hence:
\begin{align*}
OPT \leq \left(\sum_{i\in N} E\left[v_i^{1/(d-1)}\right]\right)^{\frac{d-1}{d}}
.\end{align*}
Similarly $g(x) = x^{1/(d-1)}$ is a concave function for any $d \geq 2$. 
Hence:
\begin{align*}
OPT 
\leq 
\left(\sum_{i\in N} \left(E[v_i]\right)^{1/(d-1)}\right)^{\frac{d-1}{d}} 
= 
\left(n \mu^{1/(d-1)}\right)^{\frac{d-1}{d}} 
= 
n \left(\frac{\mu}{n}\right)^{1/d}
.\end{align*}

\end{proof}

We now upper bound $OPT$ when assuming values are drawn from an MHR distribution.

\begin{theorem}
\label{thm:pseudoSurplusUbMhr}
When, $\forall i \in \bidders$, 
$\costfun[i] = \payment[i]^d$, 
where $d \ge 2$
and the value distribution $F$ is an MHR distribution:
\begin{align}
OPT
&\le
\numbidders \left(\frac{e \kappa}{n}\right)^{1/d}
.\end{align}
\end{theorem}

\begin{proof}
Lemma~\ref{lem:dhangwatnotai}
and
Lemma~\ref{lem:samuel}
tell us that: $\mu \le e R(q^*) \leq e\kappa$. 
This, combined with Lemma~\ref{cor:optUbByPseudoSurplusFordtwo},
proves the theorem.
\end{proof}

%% file: prior_free.tex
\section{Reserve Price Mechanisms}

We now turn to the design of simple prior-free mechanisms for the convex payment setting. 
We begin our analysis by looking at \emph{uniform-allocation reserve price mechanisms},  
i.e. mechanisms that allocate uniformly at random to all bidders whose value $\val[i]$ 
is above some reserve price $r$,\footnote{Or, if a good is divisible, 
uniformly across the total number of winners.} 
or equivalently, to bidders whose quantile $\quantile[i]$ is below some quantile reserve $\hat{\quantile}$. 
These mechanisms charge each of the bidders who bid at least the reserve  
$\costfun[i]^{-1} \left( \frac{\val(\hat{q})}{Z} \right)$. 
This payment rule makes the mechanism not only BIC, but ex-post IC.
%
We then describe the performance of mechanisms that select a quantile reserve uniformly at random.
Finally, we show how selecting a quantile uniformly at random 
can be easily implemented as a prior-free mechanism 
by simply picking one bidder at random from the $\numbidders$ bidders 
and using him as a reserve price setter. 
This leads to the main result of this section: 
a constant-factor approximately optimal prior-free mechanism. 
We conclude this section by noting that 
if one is allowed to make minimal use of the distribution, 
then a reserve price which is equal to the median of the distribution 
yields a better approximation ratio.  
In the Appendix, we also show that a reserve price that depends on both 
the distribution and the exponent in the payment function 
can lead to even better worst-case performance.

\paragraph{Prior-free mechanism.} 
We begin with the analysis of the revenue of a 
\emph{uniform-allocation reserve price mechanism}, 
with quantile reserve $\hat{\quantile}$.

\begin{lemma}
\label{lem:reserveQuantileApx}
Consider a convex payment setting with 
$\costfun[i] = \payment[i]^d$, $\forall i \in \bidders$, where $d \ge 1$.
Let $APX(\hat{\quantile})$ be the expected revenue of a mechanism that allocates
uniformly across all bidders with quantile $q_i\leq \hat{q}$
and charges each of these $Z$ bidders ex-post truthful payments 
$\left( \val(\hat{q}) / Z \right)^{1/d}$.
Then:
\begin{align}
APX(\hat{\quantile})
&\ge
\numbidders
\left(
\frac{\val(\hat{q})}%
{1 + (\numbidders - 1) \hat{q}}
\right)^{1/d}
\hat{q}
.\end{align}
\end{lemma}

\begin{proof}
The expected revenue the mechanism with reserve price $\val(\hat{q})$ is
\begin{align*}
APX(\hat{\quantile})
&=
\Exp_{\valvec \sim F}
\left[
\sum_{i = 1}^{\numbidders}
\left(
\frac{\val(\hat{q}) \mathbbm{1}_{\val[i] \ge \val(\hat{q})}}%
{\sum_{j = 1}^{\numbidders} \mathbbm{1}_{\val[j] \ge \val(\hat{q})}}
\right)^{1/d}
\right]
.\end{align*}
The probability that $\val[i] \ge \val(\hat{q})$ is 
$1 - F \left( \val ( \hat{q} ) \right) = \hat{q}$, 
and if there exists a winner, the denominator is at least one,
so
\begin{align*}
APX(\hat{\quantile})
&=
\Exp_{\valvec \sim F}
\left[
\sum_{i = 1}^{\numbidders}
\left(
\frac{\val(\hat{q})}%
{1 + \sum_{j = 1}^{\numbidders-2} \mathbbm{1}_{\val[j] \ge \val(\hat{q})}}
\right)^{1/d}
\right] \hat{q}
=
\sum_{i = 1}^{\numbidders}
\Exp_{\valvec \sim F}
\left[
\left(
\frac{1}%
{1 + \sum_{j = 1}^{\numbidders - 1} \mathbbm{1}_{\val[j] \ge \val(\hat{q})}}
\right)^{1/d}
\right] 
\val(\hat{q})^{1/d}
\hat{q}
.\end{align*}
The function 
$h(x) = 1/(1+x)^{1/d}$
is convex for $d \ge 1$.
By Jensen's inequality,
$\Exp [ h(x) ] \ge h (\Exp [x])$,
so we have
\begin{align*}
\Exp_{\valvec \sim F}
\left[
\left(
\frac{1}%
{1 + \sum_{j = 1}^{\numbidders - 1} \mathbbm{1}_{\val[j] \ge \val(\hat{q})}}
\right)^{1/d}
\right]
&\ge
\left(
\frac{1}%
{1 + \Exp_{\valvec \sim F}
\left[
\sum_{j = 1}^{\numbidders - 1} \mathbbm{1}_{\val[j] \ge \val(\hat{q})}
\right]}
\right)^{1/d}=
\left(
\frac{1}%
{1 + (\numbidders - 1) \hat{q}}
\right)^{1/d}
,\end{align*}
and
\begin{align*}
APX
&\ge
\sum_{i = 1}^{\numbidders}
\left(
\frac{1}%
{1 + (\numbidders - 1) \hat{q}}
\right)^{1/d}
\val(\hat{q})^{1/d}
\hat{q}
=
\numbidders
\left(
\frac{\val(\hat{q})}%
{1 + (\numbidders - 1) \hat{q}}
\right)^{1/d}
\hat{q}
.\end{align*}
\end{proof}

Given the performance of a mechanism with quantile reserve $\hat{q}$,
we now describe how well a mechanism does by selecting the quantile
reserve uniformly at random.

\begin{lemma}[Random Reserve Price Mechanism]
\label{lem:pfApx}
Consider a convex payment setting with 
$\costfun[i] = \payment[i]^d$, $\forall i \in \bidders$, where $d \ge 1$.
Let $APX$ be the expected revenue of a mechanism which draws a quantile reserve 
$\hat{q}$ uniformly at random in $[0,1]$, and then allocates
uniformly across all bidders with quantile $q_i \leq \hat{q}$, 
and charges each of these $Z$ bidders ex-post truthful payments 
$\left( \val ( \hat{q} ) / Z \right)^{1/d}$.
Then:
\begin{align}
APX\geq \frac{1}{8} 
\numbidders^{1 - 1/d}
\kappa^{1/d}
.\end{align}
\end{lemma}

\begin{proof}
A lower bound on the total expected revenue of the mechanism can be computed
by integrating $APX(\hat{q})$ with respect to quantile $\hat{q}$:
\begin{align*}
APX
&=
\int_{0}^{1} APX(\hat{q}) \, \mathrm{d}\hat{q}
\end{align*}
Invoking Lemma~\ref{lem:reserveQuantileApx}, 
and since $\hat{\quantile}\in[0,1]$, 
the quantity $APX(\hat{q})$ can be lower-bounded as follows:
\begin{align}
APX(\hat{q})
&\ge \numbidders^{1 - 1/d}
\val(\hat{q})^{1/d}
\hat{q}
.\end{align} 
Thus we get:
\begin{align*}
APX &\geq 
\numbidders^{1 - 1/d}
\int_{0}^{1}
\frac{\val(\hat{q}) \hat{q}}{\val(\hat{q})^{1 - 1/d}}
 \, \mathrm{d}\hat{q}
\ge
\numbidders^{1 - 1/d}
\int_{1/2}^{1}
\frac{R (\hat{q})}{\val(\hat{q})^{1 - 1/d}}
 \, \mathrm{d}\hat{q}
.\end{align*}
Since $\kappa = \val(1/2) \ge \val(\hat{q})$ for $1/2 \le \hat{q} \le 1$,
we have
\begin{align}
APX &\ge
\frac{\numbidders^{1 - 1/d}}{\kappa^{1 - 1/d}}
\int_{1/2}^{1}
R (\hat{q})
 \, \mathrm{d}\hat{q}
,\end{align}
and because $R (\hat{q})$ is concave, 
$\int_{1/2}^{1}
R (\hat{q})
 \, \mathrm{d}\hat{q}
\ge \frac{1}{2} \frac{1}{2} R(1/2)
= \kappa / 8$,
where the last step follows from the proof of Lemma~\ref{lem:samuel},
so
\begin{align}
APX
&\ge
\frac{1}{8} 
\numbidders^{1 - 1/d}
\kappa^{1/d}
.\end{align}
\end{proof}

We are now ready to prove the main theorem of this section: 
a constant factor prior-free mechanism. 
Observe that to draw a random quantile reserve, 
we do not need to know the distribution of values, 
as the quantile of a randomly selected bidder 
is equivalent to a randomly drawn quantile reserve. 
That is, we can sacrifice a randomly selected bidder by using his quantile as the reserve 
and run a random reserve price mechanism among the $\numbidders - 1$ bidders. 
Notice that this mechanism is prior-free.
We call this mechanism the \emph{random price setter mechanism} 
and show that it is approximately optimal.

\begin{theorem}[Prior-Free Mechanism]
Consider a convex payment setting with 
$\costfun[i] = \payment[i]^d$, $\forall i \in \bidders$, 
where $d \ge 2$ and the value distribution is an MHR distribution.
Let $APX$ be the expected revenue of the random price setter mechanism.
The random price setter mechanism achieves revenue $APX$ which satisfies:
\begin{align}
\label{eq:priorFreeRatio}
\frac{APX}{OPT}
&\ge
\frac{1}{8} 
\left(\frac{\numbidders - 1}{\numbidders}\right)^{1 - 1/d}
\frac{1}{e^{1/d}}\geq \frac{1}{16e}
.\end{align}
\end{theorem}

\begin{proof}
Observe that the revenue of the mechanism is equal to 
the revenue of the random reserve price mechanism with $\numbidders-1$ bidders. 
Thus by applying Lemma~\ref{lem:pfApx} we have:
\begin{align}
APX
&\ge
\frac{1}{8} 
(\numbidders - 1)^{1 - 1/d}
\kappa^{1/d}
.\end{align}
By Theorem~\ref{thm:pseudoSurplusUbMhr}, we also have
$OPT \le \numbidders^{1 - 1/d} (e \kappa)^{1/d}$. 
Combining the two bounds, yields the theorem.

\end{proof}

For $d \ge 2$, the approximation ratio given by Equation~\eqref{eq:priorFreeRatio} 
is $\frac{1}{8 e^{1/d}} > .075$ in the limit, as $\numbidders$ tends towards infinity,
and is $\frac{1}{8}$ as $d$ tends towards infinity as well.
We explore how well the prior-free mechanism does empirically using randomly generated
MHR distributions in Section~\ref{sec:Experiments}.

\paragraph{Detail-free mechanisms.} 
We now show that if the auctioneer has some knowledge of the distribution of values,
then mechanisms can generate more expected revenue.
Specifically, knowing the median of the distribution is sufficient 
for getting a much better approximation ratio. 
We conclude this section by providing such detail-free designs with better approximation ratios. 
The mechanisms that we propose optimizes over the quantile reserve $\hat{q}$ 
so as to get better approximation ratios.

\begin{theorem}[Median Reserve]
\label{thm:apxRatioMedianDgeTwo}
The approximation ratio of the mechanism with median reserve price $\kappa$ 
(i.e., $\hat{q} = \frac{1}{2}$) when, $\forall i \in \bidders$, 
$\costfun[i] = \payment[i]^d$, 
where $d \ge 2$ with MHR distributions, is:
\begin{align}
\label{eq:postedPriceMedianRatioDgeTwo}
\frac{APX}{OPT}
&\ge
\frac{1}{2} 
\left(\frac{2 \numbidders}{e (\numbidders + 1)}\right)^{1/d}
.\end{align}
\end{theorem}

\begin{proof}
Lemma~\ref{lem:reserveQuantileApx} tells us that when $\hat{q} = \frac{1}{2}$,
\begin{align}
APX
&\ge
\numbidders
\left(
\frac{\kappa}%
{1 + (\numbidders - 1) / 2}
\right)^{1/d}
\frac{1}{2}
=
\frac{\numbidders}{2}
\left(
\frac{2 \kappa}%
{\numbidders + 1}
\right)^{1/d}
,\end{align}
and by Theorem~\ref{thm:pseudoSurplusUbMhr},
we have:
\begin{align}
\frac{APX}{OPT}
&\ge
\frac{\numbidders}{2}
\left(
\frac{2 \kappa}%
{\numbidders + 1}
\right)^{1/d}
\frac{1}{\numbidders^{(d-1)/d} (e \kappa)^{1/d}}
=
\frac{1}{2}
\left(
\frac{2 \numbidders}%
{e(\numbidders + 1)}
\right)^{1/d}
.\end{align}
\end{proof}

Equation~\eqref{eq:postedPriceMedianRatioDgeTwo}
tells us that for large number of bidders,
the approximation ratio is $\frac{1}{2} \left(\frac{2}{e} \right)^{1/d}$.
As $d \ge 2$, this is at least $0.42$.
Further, as $d$ tends towards infinity, the approximation ratio approaches $\frac{1}{2}$.
Thus, using distribution knowledge improves performance by a factor of $4$
over the prior-free mechanism.

We now show that if we use a quantile reserve that is also 
dependent on the exponent in the payment function (rather than always the median), 
then we can further improve the approximation ratio, as is given in the following theorem 
(whose proof and analysis is deferred to Appendix~\ref{sec:cost-specific}).

\begin{theorem}[Cost-Optimized Reserve]
\label{thm:CostOptimizedReserve}
Assume that the value distributions satisfy the Monotone Hazard Rate (MHR) condition, 
and that the payment function for each bidder $i \in \bidders$
is of the form $\costfun[i] (p) = p^d$, for $d \geq 2$. 
Allocating uniformly at random to all bidders whose quantile is below 
the quantile reserve $\hat{q} = \max \left\{ \frac{1}{2}, 1 - \frac{1}{d-1} \right\}$ 
and charging them the ex-post truthful payment of $\left( \frac{v(\hat{q})}{Z} \right)^{1/d}$, 
where $Z$ is the realized number of bidders with quantile below $\hat{q}$, 
achieves revenue at least 
$\left( \frac{\numbidders}{\numbidders + 1} \right)^{1/d} \frac{1}{2\sqrt{e}}$ for $d \in [2,3)$ 
and at least 
$\left( \frac{\numbidders}{\numbidders+1} \right)^{1/d} \frac{1}{(4e(d-2))^{1/d}}$ for $d \geq 3$ 
of the optimal revenue.
\end{theorem}
Observe that for $d \rightarrow \infty$, this bound converges to $1$, 
since $(d-2)^{1/d} \rightarrow 1$. 
Thus, as the payment functions become more convex, 
the cost-optimized reserve mechanism converges to full optimality.


%% file: all_pay.tex
\section{Unknown Payment Function and All-Pay Implementation}

So far we have assumed that the auction designer knows the form of the payment function of the bidders; 
in particular, she knows the exponent $d$ in the payment function. 
This might render the auction design brittle to mis-specifications of this exponent from market knowledge. 
Hence, we explore whether there exist an auction design which is ignorant 
to the form of the payment function which achieves a constant factor approximation 
to the optimal revenue for each possible realization of the exponent $d$. 
However, we will still assume that the bidders themselves know the exponent $d$ 
and that this exponent is the same for every bidder. 

One immediate obstacle to this endeavor is that any attempt to design 
a truthful mechanism that satisfies the latter desideratum  is bound to fail.  
Any truthful mechanism will need to charge each bidder 
a payment that will make the allocation truthful.  
From Myerson's characterization of payments, such a truthful payment 
will be dependent on the exponent $d$. 
Hence, we need to resort to non-truthful auctions. 
The main idea in this section stems from the fact that, 
as was shown in Lemma~\ref{lem:feasibleAllocAB}, 
it is always optimal from a revenue maximization point of view 
to charge a deterministic payment to each bidder 
that is independent of the other bidders realized valuations. 
Such a mechanism essentially corresponds to an all-pay auction, 
i.e. an auction which solicits bids from the bidders, 
decides an allocation based on these bids, 
and then charges each bidder his bid no matter what allocation he receives. 

In this section we examine the performance of all-pay auctions 
with appropriately chosen allocation rules. 
A first attempt would be to run an all-pay implementation of the median reserve price auction: 
i.e. solicit bids, 
allocate uniformly at random to any bidder whose bid is above some carefully chosen reserve 
(so that at equilibrium it corresponds to a median reserve in the value space), 
and then charge each bidder his bid no matter what allocation he receives.
The main problem with such an attempt is that the appropriate reserve price to charge 
so that it translates to a median reserve in the value space  
is itself dependent on the exponent $d$. 
Therefore, a universal reserve price does not exist. 
Additionally, such types of non-truthful auctions where the allocation 
depends on the absolute value of the bid of a bidder 
tend to not have unique equilibria. 

Given the issues described, we look into allocation functions 
where the allocation of each bidder only depends on the relative rank of his bid 
among the bids of his opponents. 
The key idea is that we will try to approximate the median mechanism 
with such a relative ranking mechanism by  
allocating uniformly to the top half of the bidders and nothing to the bottom half. 
Assuming that bidders bid according to a symmetric Bayes-Nash equilibrium 
with a strictly monotone bid function (which will be the unique equilibrium 
based on the results of \cite{Chawla2013}), 
then at equilibrium the bidders whose valuation is among the top 50\% of valuations will be 
allocated uniformly at random an allocation of 
$2 / \numbidders$\footnote{For simplicity in this section, 
we assume that $\numbidders$ is an even number}. 
If the number of bidders is at least some constant, 
then this interim allocation function will strongly resemble the interim allocation function 
of the median mechanism: 
bidders with value strictly above the median by some margin $\epsilon$ 
will get an allocation of $2 / \numbidders$ with probability approaching $1$, 
while bidders whose value is strictly below the median by some margin $\epsilon$ 
will get an allocation that is almost $0$. 
The latter will follow from a Chernoff bound argument. 
Hence, we can show that this mechanism will achieve an approximation ratio $APX/OPT$ 
that is of the form 
$\frac{1}{\alpha} - O\left( \sqrt{\frac{\log(\numbidders)}{\numbidders}} \right)$, 
where $\alpha$ is some small constant. 
In order to avoid dependence of our analysis on a lower bound 
on the density of the distribution of values, 
we actually propose a more competitive all-pay auction where we only allocate 
to a smaller fraction of the bidders. 
In the theorem that follows we pick the top quarter, 
and we note that the constants could be further optimized by picking a more involved fraction. 
The proof of this theorem is given in Appendix~\ref{sec:app-all-pay}.\footnote{
We also note  that the constants in the analysis of this theorem could be 
optimized to obtain a better result. 
For simplicity of exposition we omit such optimization, 
as the main point of the theorem is that a constant approximation 
can be achieved by a mechanism that is independent of the payment function.}

\begin{theorem}
\label{THM:APX-ALL-PAY}
Consider an all-pay auction which allocates uniformly to the top quarter of the bidders, 
i.e.:
\begin{enumerate}
\item Solicit a bid $b_i$ from each bidder;
\item Allocate to bidder $i$ $\alloc[i] (\bid[i], \bidvec[-i] ) = \frac{4}{\numbidders}$ 
if $\bid[i]$ is among the top $\frac{\numbidders}{4}$ highest bids;
\item Charge each bidder $\bid[i]$ no matter whether they received allocation or not. 
\end{enumerate}
Assuming that $\costfun[i](\cdot)$ is a strictly monotone function that is the same for all bidders, 
and that the distribution of values is atomless and has a continuous CDF with support $[0, \bar{\val}]$, 
then at the unique Bayes-Nash equilibrium of this auction, each bidder will submit a bid $\bid[i] ( \val[i] )$ such that
\begin{equation}
\bid[i] (\val[i]) 
= 
\costfun[i]^{-1} \left( 
\val[i] \intalloc[i](\val[i]) - 
\int_{0}^{\val[i]} \intalloc[i] (\valz[i]) \, \mathrm{d}\valz[i] 
\right)
,\end{equation}
where $\intalloc[i]$ is the interim allocation that corresponds to awarding an allocation of $\frac{4}{\numbidders}$ to the top $\frac{n}{\numbidders}$ valuation bidders, i.e.
\begin{equation}
\intalloc[i](\val[i]) 
= 
\frac{4}{\numbidders} 
\Pr \left( \val[i] \text{ is among $\frac{\numbidders}{4}$ highest values} \mid \val[i] \right) 
= 
\frac{4}{\numbidders} 
\Pr \left( \sum_{j\neq i} \mathbbm{1}_{\val[j] \geq \val[i]} \leq \frac{\numbidders}{4}-1 \mid \val[i] \right)
.\end{equation}
Finally, assume that 
$\costfun[i]( \payment ) = \payment^{d}$ for some $d \geq 2$, 
$\numbidders \geq 32 \log( 16 \bar{\val} / \kappa)$, 
and that the distribution of values is MHR. 
Then this auction achieves revenue $APX$ at the unique Bayes-Nash equilibrium, which guarantees: 
\begin{equation}
\frac{APX}{OPT} \geq \frac{1}{16}
.\end{equation}
\end{theorem}

%% file: experiments.tex
\section{Experiments}
\label{sec:Experiments}

In this section, we provide empirical evidence that our proposed mechanisms
(i.e., the prior-free mechanism and the mechanisms with reserve prices)
yield near-optimal performance.
We first show that in a symmetric setting, it is possible to compute
the optimal allocation and revenue in polynomial time through the
use of interim feasibility described by~\cite{Border1991}.
Then, we describe the experimental setup and the results 
of the experiments.

\subsection{Solving for the Optimal Auction}\label{sec:optimal-auction}

We describe the mathematical programs used to solve for the optimal auction.
The implementations make use of Border's theorem
to reduce the number of allocation variables to a manageable size; 
we refer the reader to \cite{Hartline2015} for a concise description 
of interim feasibility constraints.
We present our results for continuous distributions here,
and state our result for discrete distributions,
noting that the proof for the discrete distribution setting is similar
to that of the continuous distribution.
Interested readers may find the discrete result in the appendix.

\begin{lemma}
\label{lem:borderSymContPercPay}
In a symmetric setting with convex payments,
where values are drawn from a continuous regular distribution $F$ 
with support in $T = [0,1]$ and density function $f$,
the optimal auction can be described by the following
mathematical program:
\begin{align}
\max_{z(\cdot)} &
\int_{0}^{1} f(t) 
\costfun^{-1} \left( \int_{0}^{t} \tau z(\tau) \, \mathrm{d} \tau \right) 
\, \mathrm{d} t
&
\\
\textnormal{subject to} 
&
\int_{0}^{1}
z(\tau)
\left( 1 -  F \left( \max \{ t, \tau \} \right) \right)
\, \mathrm{d} \tau
\le
\frac{1 - F(t)^{\numbidders}}{\numbidders},
& \forall t \in T
\\
&
z(\tau) \ge 0,
& \forall t \in T
\end{align}
\end{lemma}

\begin{proof}
In a symmetric environment, variables for bidder $i$ will be
equivalent to variables for bidder $j$, so 
we can simplify a complete description of an optimal auction
by focusing on one representative bidder.  
The objective function for optimal auctions is 
$\int_{0}^{1} f(t) 
\costfun^{-1} \left( \intcostfun (t) \right) 
\, \mathrm{d} t$,
which maximizes over interim allocation $\intalloc$.
By Theorem~\ref{thm:myersonIcIrPaymentFormulaMonotone},
using a slightly different form of the payment formula.
For $t \in [0,1]$, we can describe payments as
$\intcostfun (t) = \int_{0}^{t} \tau \frac{\mathrm{d} \intalloc[i] (\tau)}{\mathrm{d} \tau} \, \mathrm{d} \tau$;
that is, $z$ in the program is the derivative of the allocation function,
and we maximize over $z$.
Once a solution for $z$ is found, interim allocations can be recovered by
integrating $z$: $\intalloc (t) = \int_{0}^{t} z(\tau) \, \mathrm{d} \tau$.
The total expected revenue can be found by solving the program,
and then multiplying the objective value by the total number of bidders, $\numbidders$.

We now turn to the interim feasibility constraints.
As described by \cite{Hartline2015} (Section 8.4.1), for symmetric environments,
the only binding interim feasibility constraints say simply: for any interval of values $[t,1]$,
the expected value of the interim allocation function, given that $\val$ exceeds $t$,
must be upper bounded by the expected value of the indicator function that $\val$ is the maximum draw, again
given that $\val$ exceeds $t$:
i.e. $\Exp[\hat{x}(v) \mid v\geq t] \leq \Exp[1\{\text{$v$ is the max value}\} \mid v\geq t]$.


Taking expectations, for values ranging from $t$ to $1$, yields the following inequality:
\begin{align}
\int_{t}^{1} 
f(q) 
\int_{0}^{q}
z(\tau)
\, \mathrm{d} \tau
\, \mathrm{d} q
\le
\int_{t}^{1} f(\tau) F(\tau)^{\numbidders - 1} \, \mathrm{d} \tau,
&&
\forall t \in T
.\end{align}
Evaluating the right-hand side of the inequality,
\begin{align}
\int_{t}^{1} f(\tau) F(\tau)^{\numbidders - 1} \, \mathrm{d} \tau
=
\left. \frac{F(\tau)^{\numbidders}}{\numbidders} \right|_{t}^{1}
=
\frac{1 - F(t)^{\numbidders}}{\numbidders}
.\end{align}
We then proceed by changing the order of integration on the left-hand side of the inequality:
\begin{align}
\int_{t}^{1} 
f(q) 
\int_{0}^{q}
z(\tau)
\, \mathrm{d} \tau
\, \mathrm{d} q
&=
\int_{0}^{1}
\int_{\max \{t, \tau \}}^{1}
f(q) 
z(\tau)
\, \mathrm{d} q
\, \mathrm{d} \tau
\\
&=
\int_{0}^{1}
z(\tau)
\left. F(q) \right|_{\max \{ t, \tau \}}^{1}
\, \mathrm{d} \tau
\\
&=
\int_{0}^{1}
z(\tau)
\left( 1 -  F \left( \max \{ t, \tau \} \right) \right)
\, \mathrm{d} \tau
.\end{align}
\end{proof}

Our experiments used discrete, rather than continuous, distributions.  
Here is Lemma~\ref{lem:borderSymContPercPay}
for the discrete setting.

\begin{lemma}
\label{lem:borderSymDiscPercPay}
\label{LEM:BORDERSYMDISCPERCPAY}
In a symmetric setting with convex payments,
where values are drawn from a discrete regular distribution $F$ 
with support in $T = \{ 1, \dots, m \}$ and probability mass function $f$,
the optimal auction can be described the following
mathematical program:
\begin{align}
\max_{z} ~{} 
&
\sum_{t=1}^{m} f (t) 
\costfun^{-1} \left( \sum_{\tau = 1}^{t} \tau z_{\tau} \right) 
&
\\
\textnormal{subject to} ~{}
&
\sum_{\tau=t}^{m} z(\tau) \left( 1 - F \left( \max \{ t, \tau \} - 1 \right) \right)
\le
\sum_{q = t}^{m} f(q) y(q)
&
\\
&
z_{\tau} \ge 0,
& \forall \tau \in T
,\end{align}
where
$y(t) 
= 
\sum_{k=0}^{\numbidders - 1} 
\binom{\numbidders - 1}{k} \frac{1}{k + 1} 
\left( \sum_{\tau = 1}^{t - 1} f(\tau) \right)^{n - 1 - k} 
f(t)^k$
,
and we let
$F(0) \equiv 0$.
\end{lemma}

\subsection{Experiments and Results}

We now investigate how well our proposed mechanisms described perform empirically,
and compare the experimental results with what is theoretically guaranteed.
Along with the mechanisms described 
(Prior Free,
Posted Price (Median),
Posted Price (Monopoly),
and the optimal BIC/IIR program),
we also implemented six other mechanisms, 
which we do not provide approximation guarantees for:
a mechanism in which only one bidder is allocated (i.e., $\alloc[i] (\val[i], \valvec[-i]) \in \{ 0, 1\}$),
a mechanism in which only one bidder is allocated if she meets the monopoly reserve,
a mechanism in which all bidders of the highest value are allocated,
a mechanism in which all bidders of the highest value are allocated if they meet the monopoly reserve,
a mechanism which allocates proportionally by values,
and
a mechanism which allocates proportionally by non-negative virtual values.

To compare all these mechanisms,
we generated 50 random MHR distributions with support $T = \{ 1, \dots, 50 \}$.
For each distribution, we simulated 10,000 auctions,
and then computed the average revenue of each,
for each of $\numbidders = \{ 1, \dots, 30 \}$ symmetric bidders.
Each bidder had payment function $\costfun[i](\payment[i]) = \payment[i]^2$.
We computed the mean revenue across distributions, 
and report how well they do against the optimal BIC solution.
We summarize the results of these experiments
in Table~\ref{tab:ExperimentSummary}
and Figure~\ref{fig:ratioRev}.

Although we provide no guarantees on performance for mechanisms that allocate proportionally by virtual values, we find that this mechanism does the best among all the mechanisms we have proposed, averaging 93\% of the optimal BIC revenue for $\numbidders = 30$.  
The posted price mechanisms give similar performance, at 87\% and 88\% of OPT for the posted median price and posted monopoly price, respectively.
The mechanism that allocates proportionally by values obtains 84\% of OPT.
Finally, the prior-free mechanism performs surprisingly well, obtaining 70\% of OPT;
additional distributional knowledge may only increase total expected revenue by 30\% of OPT.

While the mechanism that allocates only to the highest bidder also does not rely on a prior,
the results highlight how important it is to have non-zero allocations in the convex payment setting.
In the usual quasi-linear utility setting, so long as the bidder with the highest virtual value is allocated, a mechanism maximizes total expected revenue.
In the convex payment setting, we see that this is no longer true.
The experimental results suggest that 
it is more desirable to allocate something to all bidders with the highest (virtual) values,
and perhaps even something to those that do not have the highest (virtual) values.
Also notice that when allocating to either only the highest value, 
or to all bidders with the highest value,
there is no benefit to allocating only to bidders that meet the monopoly reserve.
As the number of bidders grows, the probability that a bidder with a value 
larger than the monopoly reserve increases, 
so setting a reserve at the monopoly price only helps 
when few bidders are present.
Thus, we see convergence in expected revenue 
in settings where we only allocate to the highest bidder,
or to all highest bidders,
regardless of the existence of monopoly reserve prices.

For $\numbidders = 30$, by Equation~\eqref{eq:priorFreeRatio},
the prior-free mechanism has guarantee
\begin{align}
\frac{APX}{OPT}
\ge
\frac{1}{8} 
\left(\frac{29}{30}\right)^{1/2}
\left( \frac{1}{e^{1/2}} \right)
\approx 
.07
;\end{align}
the posted-price (median) mechanism, by Equation~\eqref{eq:postedPriceMedianRatioDgeTwo},
has guarantee
\begin{align}
\frac{APX}{OPT}
\ge
\frac{1}{2} 
\left(\frac{60}{31e}\right)^{1/2}
\approx
.42
;\end{align}
%
In all cases, the experimental results suggest that we may do much better than our worst-case guarantees in practice.

\begin{table}
\centering
\begin{tabular}{| l | l | l | l |}
\hline
Method & Mean Revenue & Ratio to Opt & Guarantee \\
\hline
BIC Opt  &  25.2096  &  1  &  {} \\
Prior Free*  &  17.6538  &  0.69715  &  .07 \\
Posted, Median Reserve  &  22.1199  &  0.87441  &  .42 \\
Posted, Monopoly Reserve &  22.3504  &  0.88319  &   \\ 
To Highest, No Reserve*  &  6.9006  &  0.27757  &  {} \\
To Highest, Monopoly Reserve  &  6.9006  &  0.27757  &  {} \\
To All Highest, No Reserve*  &  11.2666  &  0.44057  &  {} \\
To All Highest, Monopoly Reserve  &  11.2666  &  0.44057  &  {} \\
Prop Val*  &  21.3163  &  0.84822  &  {} \\
Prop VirVal  &  23.5587  &  0.93289  &  {} \\
\hline
\end{tabular}
\caption{Performance of the mechanisms tested for $\numbidders = 30$.
A * indicates that the mechanism does not rely on distributional knowledge.}
\label{tab:ExperimentSummary}
\end{table}

\pgfplotstableread[col sep=comma]{randDistData_1_50_nD_50_nB_30_nS_10000.csv}{\expResultsMean}
\pgfplotstableread[col sep=comma]{randDistData_1_50_nD_50_nB_30_nS_10000_ratio.csv}{\expResultsRatio}

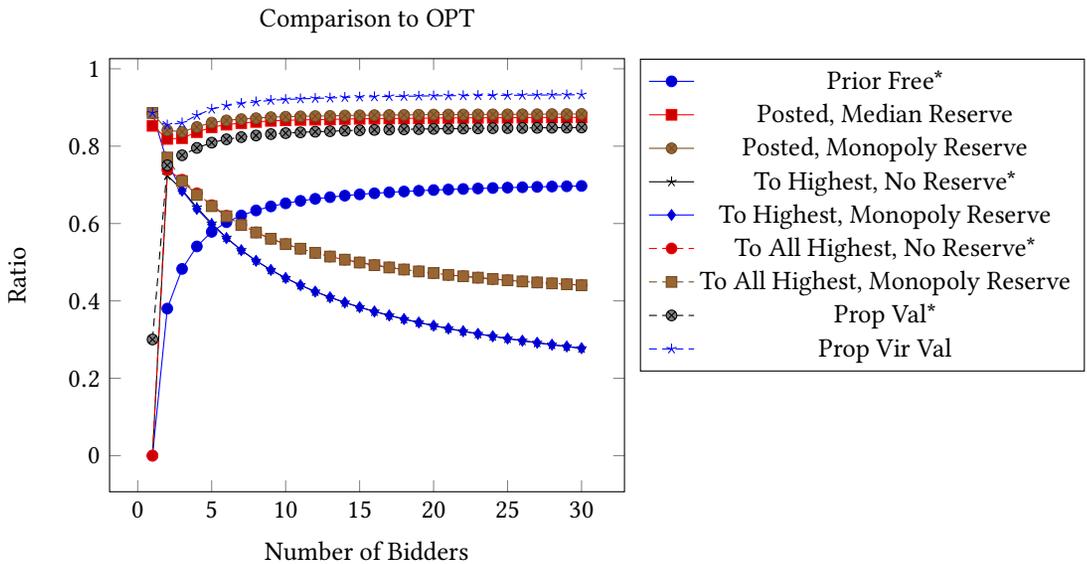
\begin{figure}[h!]
\centering
\begin{tikzpicture}
    \begin{axis}[
            xlabel=Number of Bidders,
            ylabel=Ratio,
            title={Comparison to OPT},
            legend pos=outer north east
    ]
    \addplot+[mark size=2pt] table[x = Num Bidders, y = Prior Free]{\expResultsRatio};
    \addlegendentry{Prior Free*}
    \addplot+[mark size=2pt] table[x = Num Bidders, y = Posted Median]{\expResultsRatio};
    \addlegendentry{Posted, Median Reserve}
    \addplot+[mark size=2pt] table[x = Num Bidders, y = Posted Monopoly]{\expResultsRatio};
    \addlegendentry{Posted, Monopoly Reserve}
    \addplot+[mark size=2pt] table[x = Num Bidders, y = To Highest (No Reserve)]{\expResultsRatio};
    \addlegendentry{To Highest, No Reserve*}
    \addplot+[mark size=2pt] table[x = Num Bidders, y = To Highest (Monopoly Reserve)]{\expResultsRatio};
    \addlegendentry{To Highest, Monopoly Reserve}
    \addplot+[mark size=2pt] table[x = Num Bidders, y = To All Highest (No Reserve)]{\expResultsRatio};
    \addlegendentry{To All Highest, No Reserve*}
    \addplot+[mark size=2pt] table[x = Num Bidders, y = To All Highest (Monopoly Reserve)]{\expResultsRatio};
    \addlegendentry{To All Highest, Monopoly Reserve}
    \addplot+[mark size=2pt] table[x = Num Bidders, y = ProgC Val]{\expResultsRatio};
    \addlegendentry{Prop Val*}
    \addplot+[mark size=2pt] table[x = Num Bidders, y = ProgC VirVal]{\expResultsRatio};
    \addlegendentry{Prop Vir Val}
    \end{axis}
\end{tikzpicture}
\caption{The total expected revenue compared to the optimal BIC/IIR revenue.
We can see that the performance of most mechanisms improves as we increase the number of bidders.
However, the performance of the mechanism which allocates entirely to a single bidder,
or to bidders with only the highest value, 
drops as we increase the number of bidders.
This suggests is further evidence that allocating to low values may be highly desirable in the convex payment setting.
The * indicates that the mechanism does not rely on distributional knowledge.}
\label{fig:ratioRev}
\end{figure}

%% file: conclusion.tex
\section{Conclusion}

We investigated single-parameter mechanisms where bidder's utility functions
are not linear with respect to payments, but rather convex.
Using Myerson's analysis on payments, we showed how well several mechanisms
that rely on bidder distribution knowledge can do, 
and provide performance guarantees with respect to the optimal BIC mechanism.
We additionally showed how well a prior-free mechanism can do in this setting.
The mechanisms described were empirically evaluated using valuations
drawn from random MHR distributions, 
and we saw that the total expected revenue generated by these mechanisms
can, on average, exceed the guarantees we provide.
We additionally proposed and evaluated two mechanisms that we provide no guarantees for,
which allocate proportionally to values and virtual values,
and see that they do well empirically as well,
with performance exceeding that of the mechanisms which we can provide performance guarantees for.
The theoretical results we provided describe mechanisms in which allocations are non-zero
for bidders without the largest (virtual) values.
We saw that empirically, such an allocation scheme is crucial to achieving the optimal revenue,
as a mechanism which only allocates to the highest bidder 
decreases in performance as the number of bidders decreases.
In the future, we hope to construct lower bounds on the proportional allocation mechanisms.

%% file: appendix.tex

\appendix

\input{example}

\input{single_bidder}

\section{Proof of Lemma~\ref{LEM:SOLPROGC}}
\label{sec:progCproof}
\begin{proof}
The function 
$\costfun[i]^{-1} \left( \val[i] \alloc[i] (\val[i], \valvec[-i]) \right)$
is non-decreasing and concave.
To maximize the sum 
$\sum_{i \in \bidders}
\costfun[i]^{-1} \left( \val[i] \alloc[i] (\val[i], \valvec[-i]) \right)$,
we can use the equi-marginal principle given by \cite{gossen1854entwickelung},
which states that it is optimal to allocate greedily until supply is exhausted.

The derivative of the contribution of bidder $i$ is
\begin{equation}
\label{eq:pseudoWelfareDerivativeOfBidder}
\frac{\partial \left( \val[i] \alloc[i] (\val[i], \valvec[-i]) \right)^{1/d}}
{\partial \alloc[i] (\val[i], \valvec[-i])}
= \frac{1}{d} \val[i]^{1/d} \left( \alloc[i] (\val[i], \valvec[-i]) \right)^{(1-d)/d}
.\end{equation}
Equating derivatives for bidders $i$ and $j$ (as per the equi-marginal principle), we get
\begin{equation}
\frac{1}{d} \val[i]^{1/d} \left( \alloc[i] (\val[i], \valvec[-i]) \right)^{(1-d)/d}
=
\frac{1}{d} \val[j]^{1/d} \left( \alloc[j] (\val[j], \valvec[-j]) \right)^{(1-d)/d}
.\end{equation}
The common terms can be removed, and after raising the expressions to the $d / (d - 1)$ power, 
we can simplify to
\begin{equation}
\frac{ \val[i]^{1 / (d-1)} }{\alloc[i] (\val[i], \valvec[-i])}
= \frac{ \val[i]^{1 / (d-1)} }{\alloc[j] (\val[j], \valvec[-j])}
.\end{equation}
Therefore, bidder $i$'s allocation in terms of bidder $j$ is
\begin{equation}
\alloc[i] (\val[i], \valvec[-i]) = \left( \frac{\val[i]}{\val[j]} \right)^{1 / (d - 1)} \alloc[j] (\val[j], \valvec[-j])
.\end{equation}

Plugging this expression into the ex-post feasibility condition yields:
\begin{equation}
\sum_{i=1}^{\numbidders} \alloc[i] (\val[i], \valvec[-i])
= \sum_{i=1}^{\numbidders} \left( \frac{\val[i]}{\val[j]} \right)^{1 / (d - 1)} \alloc[j] (\val[j], \valvec[-j])
.\end{equation}
It is always optimal to allocate until 
$\sum_{i=1}^{\numbidders} \alloc[i] (\val[i], \valvec[-i]) = 1$
when there is at least one bidder with a positive constant,
so the ex-post feasibility constraint can be written as
\begin{equation}
\sum_{i=1}^{\numbidders} \left( \frac{\val[i]}{\val[j]} \right)^{1 / (d - 1)} \alloc[j] (\val[j], \valvec[-j])
= 1
.\end{equation}
Therefore, bidder $j$ is allocated as follows:
\begin{equation}
\alloc[j] (\val[j], \valvec[-j]) = \frac{\val[j]^{1 / (d-1)}}{\sum_{i=1}^{\numbidders} \val[i]^{1 / (d - 1)}}
.\end{equation}
\end{proof}

\section{Improved Bound with Cost Specific Reserve}\label{sec:cost-specific}

In this section we show that if we optimize the valuation reserve as a function of $d$, then we can get a much better guarantee which converges to $1$ as $d\rightarrow \infty$. Intuitively, as $d\rightarrow \infty$ then we should be allocating to almost every player and charging them a very high price. So we should be using a reserve that allocates not with probability $1/2$ but rather with probability $q$ that is much higher than $1/2$. Thus we will use a reserve of $v(q)$ for $q\in [1/2,1]$ and optimize $q$ as a function of $d$. Specifically, we will use a $q$ of $\max\left\{\frac{1}{2},1-\frac{1}{d-1}\right\}$.

We will first prove a Lemma which is an extension of the simplified prophet inequality that we used in the previous section. 
\begin{lemma}\label{lem:general-prophet}
For any regular distribution and for any $\hat{q}\geq 1/2$: 
\begin{equation}
R(\hat{q}) \geq (1-\hat{q})R(q^*)
\end{equation}
\end{lemma}
\begin{proof}
First suppose that $q^*\geq \hat{q}$. Then we know that $R(q)$ is increasing until $\eta$ and concave. Hence:
\begin{equation}
R(\hat{q}) \geq \frac{R(q^*)-R(0)}{q^*} \hat{q} +R(0)\geq R(q^*)\hat{q} +R(0)(1-\hat{q})\geq R(q^*)\hat{q}\geq  R(q^*)\frac{1}{2} \geq R(q^*)(1-\hat{q})
\end{equation}

Now suppose that $q^*\leq \hat{q}$. Then $R(q)$ is decreasing in the region $(q^*,1]$ and $R(1)=0$. Thus by concavity of $R(q)$, since $\hat{q}\in (q^*,1]$:
\begin{equation}
R(\hat{q}) \geq \frac{R(q^*)-R(1)}{q^*-1}(\hat{q}-1)+R(1)= \frac{R(q^*)}{1-q^*}(1-\hat{q})\geq R(q^*)(1-\hat{q})
\end{equation}
\end{proof}

\begin{theorem}
Assume that the value distributions satisfy the Monotone Hazard Rate (MHR) condition and that the payment function is of the form $c(p) = p^d$ for $d\geq 2$. Then allocating uniformly at random to all players whose quantile is below the quantile reserve $\hat{q}=\max\{\frac{1}{2},1-\frac{1}{d-1}\}$ and charging them the ex-post truthful payment of $\left(\frac{v(\hat{q})}{Z}\right)^{1/d}$, where $Z$ is the realized number of players with quantile below $\hat{q}$, achieves revenue at least $\left(\frac{n}{n+1}\right)^{1/d}\cdot \frac{1}{2\sqrt{e}}$ for $d\in [2,3)$ and at least $\left(\frac{n}{n+1}\right)^{1/d}\cdot \frac{1}{(4e(d-2))^{1/d}}$ for $d\geq 3$ of the optimal revenue.
\end{theorem}
\begin{proof}
For $d\leq 3$, the reserve that we use is the median and hence the bound from Theorem \ref{thm:apxRatioMedianDgeTwo} applies. 

So we prove the bound for the case of $d\geq 3$ and $\hat{q}=1-\frac{1}{d-1}$. By Lemma \ref{lem:optUbByPseudoSurplus}, Lemma \ref{lem:dhangwatnotai} and Lemma \ref{lem:general-prophet} we have that:
\begin{equation}
OPT \leq n^{\frac{d-1}{d}} \mu^{1/d} \leq n^{\frac{d-1}{d}} (e R(q^*))^{1/d}\leq n^{\frac{d-1}{d}} \left(e\cdot \frac{R(\hat{q})}{1-\hat{q}}\right)^{1/d} = n^{\frac{d-1}{d}} \left(e\cdot \frac{\hat{q}\cdot v(\hat{q})}{1-\hat{q}}\right)^{1/d} 
\end{equation}

On the other hand by Lemma \ref{lem:reserveQuantileApx}, the revenue of the simple quantile reserve mechanism, with quantile $\hat{q}=1-\frac{1}{d-1}\geq \frac{1}{2}$ is at least:
\begin{equation}
APX \geq  \frac{n  \cdot \hat{q}\cdot v(\hat{q})^{1/d}}{\left(1+(n-1)\cdot \hat{q}\right)^{1/d}} = \frac{n  \cdot \hat{q}\cdot v(\hat{q})^{1/d}}{\left(1-\hat{q}+n\cdot \hat{q}\right)^{1/d}}\geq 
\frac{n  \cdot \hat{q}\cdot v(\hat{q})^{1/d}}{\left(\hat{q}+n\cdot \hat{q}\right)^{1/d}}=
\frac{n}{(n+1)^{1/d}}  \frac{\hat{q}\cdot v(\hat{q})^{1/d}}{\hat{q}^{1/d}}
\end{equation}
Thus the ratio of the two mechanisms is at least:
\begin{align*}
\frac{APX}{OPT} =~& \left(\frac{n}{n+1}\right)^{1/d} \frac{1}{e^{1/d}} \hat{q}^{1-2/d}(1-\hat{q})^{1/d}=
\left(\frac{n}{n+1}\right)^{1/d} \frac{1}{e^{1/d}} \left((1-\hat{q})^{d-1}\frac{1-\hat{q}}{\hat{q}}\right)^{1/d}\\
=~& \left(\frac{n}{n+1}\right)^{1/d} \frac{1}{e^{1/d}} \left(\left(1-\frac{1}{d-1}\right)^{d-1}\frac{1}{d-2}\right)^{1/d} \\
\geq~&\left(\frac{n}{n+1}\right)^{1/d} \frac{1}{e^{1/d}} \left(\frac{1}{4}\frac{1}{d-2}\right)^{1/d} 
\end{align*}
\end{proof}

\input{app_all_pay}

\section{Proof of Lemma~\ref{LEM:BORDERSYMDISCPERCPAY}}

\begin{proof}
The proof is similar to that of Lemma~\ref{lem:borderSymContPercPay}.  
In a symmetric environment, variables for bidder $i$ will be
equivalent to variables for bidder $j$, so 
we can simplify a complete description of an optimal auction
by focusing on one representative bidder.  
The objective function makes use of Theorem~\ref{thm:myersonIcIrPaymentFormulaMonotone},
using a slightly different form of the payment formula.
For $\tau \in T$, we can describe payments as
$\intcostfun (\tau) = \sum_{t=1}^{\tau} t (\intalloc (t) - \intalloc (t-1))
= \sum_{t=1}^{\tau} t z(t)$;
that is, $z(t) = \intalloc (t) - \intalloc (t - 1)$ in the program is the difference between successive allocations, where we let $\intalloc (0) \equiv 0$.
Once a solution for $z$ is found, interim allocations can be recovered by
summing $z$: $\intalloc (\tau) = \sum_{t=1}^{\tau} z(t)$.
The total expected revenue can be found by solving for the program,
and then multiplying the program value by the total number of bidders, $\numbidders$.

We now turn to the interim feasibility constraints.
As described by \cite{Hartline2015} (Section 8.4.1), for symmetric environments, the only binding interim feasibility conditions simply say that for any interval of values $[t,\infty)$ the integral of the allocation to players in that interval, must be upper bounded by the integral of the allocation to such players by the highest value wins allocation (with random tie-breaking). Under the highest value wins auction (with random tie-breaking), the expected allocation of a player with value $t$ is:
\begin{align}
y(t) 
= 
\sum_{k=0}^{\numbidders - 1} 
\binom{\numbidders - 1}{k} \frac{1}{k + 1} 
\left( \sum_{\tau = 1}^{t - 1} f(\tau) \right)^{n - 1 - k} 
f(t)^k
,\end{align} 
so
we have the following inequality:
\begin{align}
\sum_{q = t}^{m} f(t) \sum_{\tau = 1}^{q} z(\tau)
\le
\sum_{q = t}^{m} f(q) y(q)
,&&
\forall t \in T
.\end{align}
We then proceed by changing the order of summation on the left-hand side of the inequality:
\begin{align}
\sum_{q=t}^{m} f(q) \sum_{\tau = 1}^{q} z(q)
&=
\sum_{\tau=t}^{m} z(\tau) \sum_{q = \max \{t, \tau\}}^{m} f(q)
\\
&=
\sum_{\tau=t}^{m} z(\tau) \left( 1 - \sum_{k = 1}^{\max \{ t, \tau \} - 1} f(k) \right)
\\
&=
\sum_{\tau=t}^{m} z(\tau) \left( 1 - F \left( \max \{ t, \tau \} - 1 \right) \right)
.\end{align}
\end{proof}

%% file: example.tex
\section{Large Inefficiency with Highest Value Wins Auction}\label{sec:example}
Consider the following distribution of values: the value is either $1$ or $1-\epsilon$. It is equal to $1$ with probability $\frac{\log(n)}{\sqrt{n}}$ and to $1-\epsilon$ with probability $1-\frac{\log(n)}{\sqrt{n}}$. 

We first present a rough sketch of the proof. As the number of players grows large, then with very high probability we will have approximately $\sqrt{n}$ players that have value $1$ and $n-\sqrt{n}$ players who have value $1-\epsilon$. So the interim allocation of the highest bidder wins auction is approximately: $1/\sqrt(n)$ for a player with value $1$ and $0$ for a player with value $1-epsilon$. Moreover, the payment of a player with value $1$ is at most $\sqrt{x(1)}\approx 1/N^{1/4}$, while the payment of a player with value $1-\epsilon$ is approximately $0$. Thus the expected payment of a single bidder is approximately $\frac{1}{\sqrt{n}} * \frac{1}{N^{1/4}}$ and the total expected payment is $N$ times that which is $n^{1/4}$.  On the other hand, if we allocate to everyone uniformly at random (as long as they bid at least $1-\epsilon$), then we can charge everyone $\sqrt{\frac{1-\epsilon}{n}}$, leading to a total revenue of $\sqrt{n(1-\epsilon)}$. As $\epsilon\rightarrow 0$, the ratio of the optimal mechanism and the mechanism which allocates to the highest bidder (breaking ties uniformly at random) is $O(n^{1/4})$.

We now show more formally that these asymptotics hold. In particular we show that the revenue of the highest bidder wins auction is of $O(n^{1/4})$. The proof then follows by the argument above. Let $X_i$ denote the $\{0,1\}$ indicator random variable of whether a player $i$ has value $1$ or not. Then the interim allocation of a player with value $1$ is:
\begin{equation}
x(1) = \Exp\left[ \frac{1}{1+\sum_{j\neq i} X_j}\right]\leq \Exp\left[ \frac{1}{\sum_{j} X_j}\right]
\end{equation} 
Since $\sum_{j} X_j$ is the sum of $n$ independent $\{0,1\}$ random variables which are $1$ with probability $1/\sqrt{n}$, we have that $\frac{1}{n}\sum_{j}X_j$ is at least $\frac{\log(n)}{\sqrt{n}}-\frac{\log(2/\delta)}{\sqrt{n}}$ with probability at least $1-\delta$, by applying the Chernoff bound. Thus we have:
\begin{equation}
x(1) \leq \frac{1}{(\log(n)-\log(1/\delta))\sqrt{n}}(1-\delta)+1\cdot \delta\leq \frac{1}{(\log(n)-\log(1/\delta))\sqrt{n}}+\delta 
\end{equation}
By setting $\delta = \frac{1}{\sqrt{n}}$, the latter is upper bounded by:
\begin{equation}
x(1) \leq \frac{1}{(\log(n)-\log(1/\delta))\sqrt{n}}+\delta  = \frac{2}{\log(n)\sqrt{n}}+\frac{1}{\sqrt{n}} \leq \frac{3}{\sqrt{n}}
\end{equation}

Now we also argue about the expected allocation of players with value $1-\epsilon$. Such a player is allocated only when every other player has value $1-\epsilon$ and in that case he is allocated uniformly at random, hence getting allocation $1/n$. Thus his expected allocation is:
\begin{equation}
x(0) = \frac{1}{n} \left(1-\frac{\log(n)}{\sqrt{n}}\right)^{n-1}\leq \frac{1}{n}
\frac{1}{n} \left(1-\frac{1}{\sqrt{n-1}}\right)^{n-1} \leq \frac{1}{n} \left(\frac{1}{e}\right)^{\sqrt{n-1}} \leq \frac{1}{n^2}
\end{equation}
since $e^{\sqrt{x-1}}>x$ for any $x>1$.

So the total expected payment collected by a single player is at most:
\begin{equation}
\frac{\log(n)}{\sqrt{n}}\sqrt{x(1)} + \sqrt{x(1-\epsilon)}\leq \frac{2\log(n)}{\sqrt{n}n^{1/4}}+\frac{1}{n}\leq \frac{3\log(n)}{n^{3/4}}
\end{equation}
Leading to a total expected payment of at most $3n^{1/4} \log(n)$.

%% file: single_bidder.tex
\section{The Single Bidder Setting}

Here, we describe the performance of a mechanism where there is only one bidder,
where the reserve price is set to the median of the bidder's distribution, $\kappa$.

\begin{theorem}
\label{thm:apxRatioMedianSingleBidder}
The approximation ratio of the mechanism with median reserve price $\kappa$
in a single bidder environment where the bidder has
a convex payment function $\costfun[i]$, 
and valuations are drawn from a regular distribution is $\frac{1}{2}$.
\end{theorem}

\begin{proof}
From Corollary~\ref{cor:myersonPayment}, we can describe the expected revenue
of the mechanism with reserve price $\kappa$ as 
$\Exp_{\val[i] \sim F_i} \left[ \intpayment(\val[i]) \right]
=
\Exp_{\val[i] \sim F_i} \left[ \costfun[i]^{-1} (h (\val[i])) \right]
$,
where $h^*(\val[i]) =  \costfun[i] (\intpayment[i](\val[i]))$ is the 
revenue-maximizing interim convex payment function for 
revenue-maximizing payment $\intpayment[i]^* (\val[i])$, and 
$\costfun[i]^{-1}$ is the inverse of the convex payment function.
By Jensen's inequality, the optimal revenue is upper-bounded by
\begin{align}
OPT 
\le 
\costfun[i]^{-1} \left( \Exp_{\val[i] \sim F_i} \left[ h_i^* (\val[i])\right] \right)
.\end{align}

By Lemma~\ref{lem:samuel}:, $R(q^*) \le \kappa$.  
By definition of the monopoly reserve price,
$\eta \left( 1 - F (\eta) \right) = R(q^*)$.
Therefore, $\eta \left( 1 - F (\eta) \right) \le \kappa$.
Furthermore, by Myerson's analysis of virtual values (Theorem~\ref{thm:myersonVirVal}),
we know that maximizing total expected virtual surplus maximizes 
the total expected payments,
which is equivalent to maximizing the revenue curve:
\begin{align}
OPT 
&\le 
\costfun[i]^{-1} \left( R(q^*) \right)
\\
&\le
\costfun[i]^{-1} \left( \kappa \right)
.\end{align}
A mechanism that allocates to a bidder that bids at least $\kappa$ and charges $\kappa$ 
yields revenue:
\begin{align}
APX 
&= 
\costfun[i]^{-1} (\kappa) (1 - F_i(\kappa))
\\
&=
\frac{1}{2} \costfun[i]^{-1} (\kappa)
.\end{align}
Therefore, the approximation ratio of this mechanism is
\begin{align}
\frac{APX}{OPT}
&\ge
\frac{1}{2} \costfun[i]^{-1} (\kappa) \frac{1}{\costfun[i]^{-1} \left( \kappa \right)}
\\
&=
\frac{1}{2}
.\end{align}
\end{proof}

%% file: app_all_pay.tex
\section{Proof of Theorem \ref{THM:APX-ALL-PAY}}\label{sec:app-all-pay}

The fact that the given bidding function is an equilibrium follows from the fact that it is a monotone function of the valuation of each player. Therefore, it gives rise to the interim allocation described in the theorem. Moreover, the pair of the bid function and interim allocation satisfy Myerson's payment identity. Finally, no player wants to bid outside of the support of the bid distribution. Hence, from standard analysis that can be found in \cite{Hartline2015}, this implies that the proposed pair of a bid function and an interim allocation constitute a Bayes-Nash equilibrium. Uniqueness of this equilibrium follow from the results of \cite{Chawla2013}.

We now move on to analyzing the revenue of this equilibrium of the all-pay auction. For simplicity we will assume that the number of bidders is a multiple of $4$. The result easily extends to the general case, albeit with more complex notation. 

We first prove upper and lower bounds on the interim allocation of bidders as a function of their quantile $q_i=1-F(\val[i])$. For any quantile $q$, denote with $X_j(q)=1\{q_j \leq q\}$. Observe that: $\Exp[X_j(q)]=\Pr[q_j\leq q]=q$, since quantiles are distributed uniformly in $[0,1]$. Moreover, if we denote with $\intalloc[i](q_i)$ the interim allocation of player $i$ when he has quantile $q_i$, then:
\begin{equation}
\intalloc[i](q_i) = \frac{4}{n} \Pr\left[ \sum_{j\neq i} X_j(q_i)\leq \frac{n}{4}-1\right]
\end{equation}

Let $S_{n-1}(q_i) = \sum_{j\neq i} X_j(q_i)$. Since $S_{n-1}$ is the sum of $n-1$ independent $0/1$ random variables, each with success probability $q$, we get by the Chernoff bound that, for any $\epsilon>0$:
\begin{equation}
\Pr[ (q_i-\epsilon)(n-1)\leq S_{n-1}(q_i) \leq (q_i+\epsilon)(n-1)] \geq 1-2\exp\{-2\epsilon^2 n\}  
\end{equation}

Thus we have that for $q_i$, such that: $(q_i+\epsilon)(n-1)\leq \frac{n}{4}-1$:
\begin{align*}
\intalloc[i](q_i)=~&\frac{4}{n}\Pr[ S_{n-1}(q_i) \leq  \frac{n}{4}-1]=\frac{4}{n}\left(1-\Pr[ S_{n-1}(q_i) > \frac{n}{4}-1] \right)\\
\geq~&\frac{4}{n}\left(1-\Pr[ S_{n-1}(q_i) > (q_i+\epsilon)(n-1)] \right)\geq \frac{4}{n}\left(1-2\exp\{-2\epsilon^2 n\}\right)
\end{align*}
Re-arranging the condition on $q_i$, we get that the latter bound on the interim allocation holds for $q_i\leq \frac{1}{4} - \frac{3}{4}\frac{1}{n-1}-\epsilon$.

Similarly, we have that for $q_i$, such that: $(q_i-\epsilon)(n-1)> \frac{n}{4}-1$:
\begin{align*}
\intalloc[i](q_i)=~&\frac{4}{n}\Pr[ S_{n-1}(q_i) \leq \frac{n}{4}-1] \\
\leq~&\frac{4}{n}\Pr[ S_{n-1}(q_i) \leq (q_i-\epsilon)(n-1)] \leq \frac{8}{n}\exp\{-2\epsilon^2 n\} 
\end{align*}
Re-arranging the condition on $q_i$, we get that the latter bound on the interim allocation holds for $q_i> \frac{1}{4} - \frac{3}{4}\frac{1}{n-1}+\epsilon$.

Finally, we know that the interim allocation $\intalloc[i](q_i)$ is monotone decreasing in $q_i$.

Now we consider the perceived payment of a player as a function of his quantile, which by transforming Myerson's identity to quantile space, takes the following form:
\begin{align*}
c_i(b_i(q_i)) = \val[i](q_i) \intalloc[i](q_i) + \int_{q_i}^{1} \intalloc[i](z)\val[i]'(z)dz=\val[i](q_i) \intalloc[i](q_i) - \int_{q_i}^{1} \intalloc[i](z)|\val[i]'(z)|dz
\end{align*}
We will lower bound the perceived payment of a player with quantile $q_i\leq \frac{1}{4} - \frac{3}{4}\frac{1}{n-1}-\epsilon$:
\begin{align*}
c_i(b_i(q_i)) =~& \val[i](q_i) \intalloc[i](q_i) + \int_{q_i}^{} \intalloc[i](z)\val[i]'(z)dz\\
\geq~& \val[i](q_i)\intalloc[i](q_i) - \int_{q_i}^{\frac{1}{4} - \frac{3}{4}\frac{1}{n-1}+\epsilon} \intalloc[i](z)|\val[i]'(z)|dz-\int_{\frac{1}{4} - \frac{3}{4}\frac{1}{n-1}+\epsilon}^1 \intalloc[i](z)|\val[i]'(z)|dz\\
\geq~& \val[i](q_i) \intalloc[i](q_i)- \intalloc[i](q_i)\int_{q_i}^{\frac{1}{4} - \frac{3}{4}\frac{1}{n-1}+\epsilon}|\val[i]'(z)|dz-\frac{8}{n}\exp\{-2\epsilon^2n\}\int_{\frac{1}{4} - \frac{3}{4}\frac{1}{n-1}+\epsilon}^1 |\val[i]'(z)|dz\\
\geq~& \val[i]\left(\frac{1}{4} - \frac{3}{4}\frac{1}{n-1}+\epsilon\right) \intalloc[i](q_i)-\frac{8}{n}\exp\{-2\epsilon^2n\}\int_{\frac{1}{4} - \frac{3}{4}\frac{1}{n-1}+\epsilon}^1 |\val[i]'(z)|dz\\
\geq~& \val[i]\left(\frac{1}{4} - \frac{3}{4}\frac{1}{n-1}+\epsilon\right) \frac{4}{n}(1-2\exp\{-2\epsilon^2 n\})-\frac{8}{n}\exp\{-2\epsilon^2n\}\bar{\val}\\
\geq~& \frac{4}{n}\val[i]\left(\frac{1}{4}+\epsilon\right) - \frac{16}{n}\exp\{-2\epsilon^2 n\}\bar{\val}
\end{align*}
 
Since a player has quantile smaller than $\frac{1}{4} - \frac{3}{4}\frac{1}{n-1}-\epsilon$ with probability equal to $\frac{1}{4} - \frac{3}{4}\frac{1}{n-1}-\epsilon)$, we can lower bound the ex-ante expected payment of each bidder by:
\begin{align*}
\Exp[b_i(q_i)]\geq \left(\frac{1}{4} - \frac{3}{4}\frac{1}{n-1}-\epsilon\right)\cdot c_i^{-1}\left(\frac{4}{n}\val[i]\left(\frac{1}{4}+\epsilon\right) - \frac{16}{n}\exp\{-2\epsilon^2 n\})\bar{\val}\right)
\end{align*}
By picking $\epsilon = \sqrt{\frac{\log(16\bar{\val}/\kappa)}{2n}}$ (where $\kappa=v(1/2)$ is the median), then we get:
\begin{align*}
\Exp[b_i(q_i)]\geq \left(\frac{1}{4} - \frac{3}{4}\frac{1}{n-1}-\sqrt{\frac{\log(16\bar{\val}/\kappa)}{2n}}\right)\cdot c_i^{-1}\left(\frac{4}{n}\val[i]\left(\frac{1}{4}+\sqrt{\frac{\log(\bar{\val}/\kappa)}{2n}}\right) - \frac{1}{n}\kappa\right)
\end{align*}
Assuming that $\sqrt{\frac{\log(16\bar{\val}/\kappa)}{2n}}\leq \frac{1}{8}$, which happens when $n\geq 32 \log(16\bar{\val}/\kappa)$, then we have that:
\begin{equation}
\val[i]\left(\frac{1}{4}+\sqrt{\frac{\log(\bar{\val}/\kappa)}{2n}}\right) \geq \val[i]\left(1/2\right) = \kappa
\end{equation}
which yields the lower bound:
\begin{align*}
\Exp[b_i(q_i)]\geq \left(\frac{1}{8} - \frac{3}{4}\frac{1}{n-1}\right)\cdot c_i^{-1}\left(\frac{3}{n}\kappa\right)
\end{align*}
Further assuming that $\frac{3}{4}\frac{1}{n-1}\leq \frac{1}{16}\Leftrightarrow n\geq 13$ (which holds whenever $n\geq 32 \log(16\bar{\val}/\kappa)$), we get that:
\begin{align*}
\Exp[b_i(q_i)]\geq \frac{1}{16}\cdot c_i^{-1}\left(\frac{3}{n}\kappa\right)
\end{align*}

Combining all of the above we get that for $n\geq 32 \log(16\bar{\val}/\kappa)$, we can lower bound the revenue of the all-pay auction by:
\begin{equation}
APX = n \cdot \frac{1}{16}\cdot c_i^{-1}\left(\frac{3}{n}\kappa\right)
\end{equation}

If $c_i(x) = x^d$, then:
\begin{equation}
APX \geq  n \cdot \frac{1}{16}\cdot \left(\frac{3}{n}\kappa\right)^{1/d}
\end{equation}
On the other hand by Lemma \ref{lem:optUbByPseudoSurplus}, we have that for MHR distributions:
\begin{equation}
OPT\leq \numbidders  \left(\frac{e \kappa}{n}\right)^{1/d}
\end{equation}
Hence, we conclude the final part of the theorem that:
\begin{equation}
\frac{APX}{OPT} \geq \frac{1}{16} \left(\frac{3}{e}\right)^{1/d} \geq \frac{1}{16}
\end{equation}